\documentclass[format=sigconf,nonacm,protrusion=true,expansion=true]{acmart}

\usepackage[utf8]{inputenc}
\usepackage{fullpage}
\usepackage{multicol}
\usepackage{cleveref}
\usepackage{xspace}
\usepackage{graphicx}

\usepackage[multiple]{footmisc}

\usepackage{enumitem}
\setlist[itemize]{leftmargin=*}
\setlist[enumerate]{leftmargin=*}
\setlist[description]{leftmargin=1em}

\usepackage{tikz}
\usetikzlibrary{fit,shapes,decorations.pathmorphing}
\usepackage{tikz-cd}
\usepackage[textsize=tiny]{todonotes}

\newcommand{\define}[1]{\textit{#1}}

\newcommand{\Crown}[1]{\mathfrak{C}{#1}}

\usepackage[ruled,vlined]{algorithm2e}
\crefname{algocf}{alg.}{algs.}
\Crefname{algocf}{Algorithm}{Algorithms}

\newtheorem{theorem}{Theorem}
\newtheorem{lemma}[theorem]{Lemma}
\newtheorem{corollary}[theorem]{Corollary}
\theoremstyle{definition}

\newtheorem{definition}[theorem]{Definition}
\theoremstyle{remark}
\newtheorem{remark}[theorem]{Remark}
\newtheorem{example}[theorem]{Example}

\newcommand{\addtikzpadding}{\useasboundingbox ([shift={(0mm,1mm)}]current bounding box.north east) rectangle ([shift={(0mm,0mm)}]current bounding box.south west)}

\newcommand\blfootnote[1]{%
  \begingroup
  \renewcommand\thefootnote{}\footnote{#1}%
  \addtocounter{footnote}{-1}%
  \endgroup
}

\title[A distributed blossom algorithm for minimum-weight perfect matching]{A distributed blossom algorithm\\for minimum-weight perfect matching}
\author[Eric C.\ Peterson]{Eric C.\ Peterson${}^*$}
\affiliation{
\institution{Eigenware}
\city{Berkeley}
\state{CA} 
\country{USA}}
\email{peterson.eric.c@gmail.com}
\orcid{0000-0002-1633-0050}
\author[Peter J.\ Karalekas]{Peter J.\ Karalekas${}^*$}
\affiliation{
\institution{Eigenware}
\city{Berkeley}
\state{CA} 
\country{USA}}
\email{peter@karalekas.com}
\orcid{0000-0001-6031-4800}

\begin{abstract}
We describe a distributed, asynchronous variant of Edmonds's exact algorithm for producing perfect matchings of minimum weight~\cite{EdmondsMaximumMatching01}.
The development of this algorithm is driven by an application to online error correction in quantum computing, first envisioned by Fowler~\cite{FWH}; we analyze the performance of our algorithm as applied to this domain in a sequel~\cite{PetersonKaralekasArtemisPaper}.
\end{abstract}

\begin{document}

\maketitle
\section{Introduction}

% Two aims of distributed computing are to accelerate computation by performing several independent computational steps in parallel, and to alleviate the resource requirements of an individual computational actor by dividing that responsibility among several others.
% These two factors are in contention: the more state shared by actors (or governed by a particular actor), the more difficult it is to parallelize operations; and the less state shared by actors, the narrower the space of possible computations.
% Distributed algorithm design studies how these two poles can be balanced and whether there are any gains to be made in the exchange.

Edmonds's now-classic results on maximum matchings lie at the intersection of computer science, combinatorics, and integer linear programming: starting from a known polynomial-time algorithm for producing maximum matchings in bipartite graphs~\cite[Proposition 5.7]{CCPS}, he showed first that a polynomial-time modification could be used to handle a non-bipartite graph~\cite{EdmondsPathsTreesFlowers}, then that in the presence of edge weights another polynomial-time modification could be used to produce a minimum-weight representative among maximum matchings~\cite{EdmondsMaximumMatching01}.
These are respectively called the ``blossom algorithm'' and the ``weighted blossom algorithm''.
These landmark results set in motion broad research programs in several domains: there are theoretical consequences in both computer science and mathematics; the algorithmic technique itself admits both generalizations and efficiency improvements; and it opened the door to a host of applications.

As an example of such an application, minimum-weight perfect matchings (MWPMs) attracted the attention of quantum computer scientists, who showed that an MWPM solver can be used as an approximation algorithm for decoding syndromes appearing in quantum error correction, with approximation ratio dependent on the physical properties of the underlying quantum device~\cite{DKLP}.
This idea has taken such hold with designers of quantum computers that it has appeared in a variety of surveys on the subject (see, e.g.,~\cite{FMM,CNAACHIPBFKRPJSNPB}) as a solved problem.
However, in order to deploy this on a live quantum device, Fowler showed that one must make use of a ``parallelized'' MWPM solver~\cite{FWH}, and work has stopped short of producing (or referencing) such an algorithm.
\blfootnote{${}^*$ All work was done prior to joining Amazon Web Services.}

Careful consideration of the intended application indicates that the ``parallelized'' implementation must actually be \emph{distributed} with only local information available to each worker, \emph{online} so as to cope with a dynamic problem graph, and ideally \emph{asynchronous} to best match lab hardware.
Meanwhile, though state of the art in MWPM solvers has advanced substantially since the '60s, they have had other concerns top of mind: it is easy to show that the worst-case complexity of an exact solution to the matching problem on a cycle graph has runtime polynomial in the diameter~\cite{Linial}, which has encouraged the development of approximation algorithms instead~(\cite{WattenhoferWattenhofer}, \cite{LPSR}, \cite{LPSP}, and many others).

In this paper, driven by the extra structure available on the problem graphs in our intended application, we return to the exact setting: we describe an exact, asynchronous distributed blossom algorithm suitable for fulfilling Fowler's claim and prove its correctness.
As part of extending Edmonds's algorithm to operate on alternating forests rather than trees, we draw the reader's attention to a new, naturally-occurring forest operation which we call \define{multireweight}, which does not arise during serial execution and which is crucial to the correctness of the distributed algorithm.
We also provide an implementation of the algorithm, \texttt{anatevka}~\cite{anatevka,Aleichem}, as part of a simulation testbed for distributed systems described in a previous paper~\cite{PetersonKaralekasAetherPaper,aether}.
We make no reference to quantum computing outside of this introduction, since the existence and behavior of this algorithm is entirely a matter of distributed computing.
Instead, we direct interested readers to the sequel paper~\cite{PetersonKaralekasArtemisPaper} for the further modifications necessary to the application and the performance analysis in that context.

% \todo[inline]{%
% Peter suggested that we make contact with other modern literature on distributed matchers, out of which he only found approximation algorithms.
% Eric thinks that should go either here, in the introduction, or in the Closing Comments section.
% }

% \todo[inline]{%
% \textbf{Old todo from Peter:} when introducing the blossom algorithm, we might consider trying to reference some other practical use cases beyond QC -- in doing a brief search, it looks like there are potentially VLSI routing applications
% }

\section{The serial algorithm}\label{ClassicalSection}

Our distributed algorithm is most easily cast as a piecewise modification of Edmonds's serial blossom algorithm, with one extra operation.
To facilitate such a description, and to put the unfamiliar reader at ease, we first review the details of the serial algorithm.
The inputs and output of the problem are:

\begin{definition}
A \define{matching} $M$ on a graph $G$ is a set of edges with no repeated vertices.
A matching is \define{maximum} when it is of maximum cardinality.
A maximum matching on $G$ is \define{perfect} when $G$ has an even number of vertices.
For edge-weighted $G$, a matching is said to be \define{minimum weight} if there is no equal-sized matching with smaller edge weight sum.
\end{definition}

\noindent
The goal is to produce minimum-weight perfect matchings.

\begin{remark}[Standing assumptions on $G$]
Initially, we will assume $G$ to be unweighted and bipartite, though we will drop these assumptions as our discussion progresses.
Between any pair of vertices in $G$, we permit there to be no edges, one edge, or several edges---but since a loop can never be a match edge, it is harmless to assume that $G$ is loopless.
For the purposes of our description, it is convenient to permit the case of multiple edges, but it is not necessary: the algorithm will behave as if there is at most one edge between any pair of vertices (viz., the one of least weight).
In the weighted setting, one can also model missing edges by edges of infinite weight, so that the entire algorithm need only be described for a complete graph.
\end{remark}

\subsection{Augment and graft}

\noindent
Suppose that we are given a matching $M$, perhaps not yet maximum.
The core mechanism of the algorithm is to identify an augmenting path:

\begin{definition}
We say that an edge is \define{matched} if it belongs to $M$ or otherwise that it is \define{unmatched}, and we say that a vertex is \define{matched} if it is the endpoint of any matched edge or otherwise that it is unmatched.
Thus, an \define{alternating chain} or \define{alternating path} is a sequence of adjoining edges which alternate between being matched and unmatched.
An \define{augmenting path} is an alternating path whose first and last vertices are both unmatched.
\end{definition}

\noindent
With such a path in hand, one can produce a new matching by inverting which edges in the path belong to the matching. The number of edges in the new matching is one larger than that of the old.

\begin{definition}
This inversion procedure is called \define{augmenting} $M$ along the augmenting path.

\end{definition}

\begin{example}
See \Cref{AugmentPathFigure} for a depiction of augmentation.
\end{example}

\begin{figure}
\begin{tikzpicture}[main/.style = {draw, circle}, node distance=12mm]
\node[main] (r) {$r$}; 
\node[main] (s) [right of=r] {$s$};
\node[main] (t) [right of=s] {$t$};
\node[main] (u) [right of=t] {$u$};
\node[main] (v) [right of=u] {$v$};
\node[main] (w) [right of=v] {$w$};
\draw (r) -- (s);
\draw[very thick] (s) -- (t);
\draw (t) -- (u);
\draw[very thick] (u) -- (v);
\draw (v) -- (w);
\addtikzpadding;
\end{tikzpicture}
\[\Downarrow\]
\begin{tikzpicture}[main/.style = {draw, circle}, node distance=12mm]
\node[main] (r) {$r$}; 
\node[main] (s) [right of=r] {$s$};
\node[main] (t) [right of=s] {$t$};
\node[main] (u) [right of=t] {$u$};
\node[main] (v) [right of=u] {$v$};
\node[main] (w) [right of=v] {$w$};
\draw[very thick] (r) -- (s);
\draw (s) -- (t);
\draw[very thick] (t) -- (u);
\draw (u) -- (v);
\draw[very thick] (v) -- (w);
\end{tikzpicture}
\caption{%
The graph on top illustrates an augmenting path joining $r$ to $w$: neither $r$ nor $w$ is matched, and the edges between them alternate between not belonging and belonging to the matching.
The graph below shows the effect of augmenting along this path: whether an edge is or is not a member of the matching reverses, and the size of the matching increases by one edge.
}
\label{AugmentPathFigure}
\end{figure}
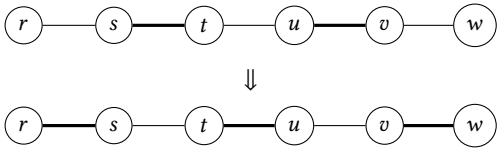

\noindent
The meat of the algorithm is then a search for augmenting paths, which begin and end at unmatched vertices.
The data structure which powers this is an \define{alternating tree}:

\begin{definition}
An \define{alternating tree} is a tree whose root is unmatched, and whose edges alternate between unmatched and matched as they descend from the root.
We refer to the even- and odd-depth tree vertices respectively as \define{positive} and \define{negative}, and we write $T_+$ and $T_-$ for these subsets of vertices.%
\footnote{%
Some authors refer to positive, negative, and unmatched vertices respectively as \define{outer}, \define{inner}, and \define{exposed}.
}
\end{definition}

\begin{definition}
The inductive operation used to assemble such an alternating tree is called \define{grafting}.\footnote{%
Some authors call this operation \define{grow}~\cite{Kolmogorov}.
}
Let $M$ be an intermediate matching, and let $T$ be an alternating subtree of the ambient graph $G$.
Select a pair of edges $e$ and $f$, as in
\begin{center}
\begin{tikzpicture}[main/.style = {draw, circle}, node distance=15mm]
\node[main] (u) {$u$};
\node[main] (v) [right of=u] {$v$};
\node[main] (w) [right of=v] {$w$};
\draw[densely dotted] (u) -- (v) node[midway, fill=white]{$e$};
\draw[very thick] (v) -- (w) node[midway, fill=white]{$f$};
\addtikzpadding;
\end{tikzpicture}
,
\end{center}
with the additional properties that
\begin{enumerate}
    \item $u$ belongs to $T$, but $v$ and $w$ do not.
    \item If $u$ has a parent in $T$, the edge to that parent is in $M$.
    \item $f$ belongs to $M$, but $e$ does not.
\end{enumerate}
We then define the \define{graft} of this edge pair onto $T$ to be the union $T' = T \cup \{e, f\}$.
The first property of the edge pair ensures that $T'$ is a tree, and the others ensure that $T'$ is alternating.
\end{definition}

% \begin{example}
% We illustrate grafting in \Cref{GraftFigure}.
% \end{example}

\begin{figure}
\begin{tikzpicture}[main/.style = {draw, circle}, node distance=15mm]
\node[main] (r) {$r$};
\node[main] (s) [right of=r] {$s$};
\node[main] (t) [right of=s] {$t$};

\node[main] (u) [above right of=t] {$u$};
\node[main] (v) [right of=u] {$v$};

\node[main] (w) [below right of=t] {$w$};
\node[main] (x) [right of=w] {$x$};

\node[main] (q) [right of=v] {$q$};

\draw[->] (r) -- (s);
\draw[->, very thick] (s) -- (t);
\draw[red, -, densely dotted] (t) -- (u) node[midway, fill=white] {$e$};
\draw[red, -, very thick] (u) -- (v) node[midway, fill=white] {$f$};
\draw[->] (t) -- (w);
\draw[->, very thick] (w) -- (x);

\draw[densely dotted] (v) -- (q) node[midway, fill=white] {$g$};
\addtikzpadding;
\end{tikzpicture}
\[\Downarrow\]
\begin{tikzpicture}[main/.style = {draw, circle}, node distance=15mm]
\node[main] (r) {$r$};
\node[main] (s) [right of=r] {$s$};
\node[main] (t) [right of=s] {$t$};

\node[main] (u) [above right of=t] {$u$};
\node[main] (v) [right of=u] {$v$};

\node[main] (w) [below right of=t] {$w$};
\node[main] (x) [right of=w] {$x$};

\node[main] (q) [right of=v] {$q$};

\draw[red, ->] (r) -- (s);
\draw[red, ->, very thick] (s) -- (t);
\draw[red, ->] (t) -- (u) node[midway, fill=white] {$e$};
\draw[red, ->, very thick] (u) -- (v) node[midway, fill=white] {$f$};
\draw[->] (t) -- (w);
\draw[->, very thick] (w) -- (x);

\draw[red, densely dotted] (v) -- (q) node[midway, fill=white] {$g$};
\end{tikzpicture}
\[\Downarrow\]
\begin{tikzpicture}[main/.style = {draw, circle}, node distance=15mm]
\node[main] (r) {$r$};
\node[main] (s) [right of=r] {$s$};
\node[main] (t) [right of=s] {$t$};

\node[main] (u) [above right of=t] {$u$};
\node[main] (v) [right of=u] {$v$};

\node[main] (w) [below right of=t] {$w$};
\node[main] (x) [right of=w] {$x$};

\node[main] (q) [right of=v] {$q$};

\draw[red, -,very thick] (r) -- (s);
\draw[red, densely dotted] (s) -- (t);
\draw[red, -, very thick] (t) -- (u) node[midway, fill=white] {$e$};
\draw[red, densely dotted] (u) -- (v) node[midway, fill=white] {$f$};
\draw[densely dotted] (t) -- (w);
\draw[-, very thick] (w) -- (x);

\draw[red, -, very thick] (v) -- (q) node[midway, fill=white] {$g$};
\end{tikzpicture}
\caption{%
Begin by grafting a matched edge $f$ along edge $e$ onto an alternating tree $T$ rooted at $r$. This creates an augmenting path (middle, red) formed from a branch of $T$ and an edge $g$ not in $T$. Then augment through this path to produce a maximum matching.
Note that edges participating in the tree carry arrow heads (pointing toward the leaves), edges participating in the partial matching are bold, and the remaining edges are dotted.
}
\label{AugmentAndGraftFigure}
\end{figure}
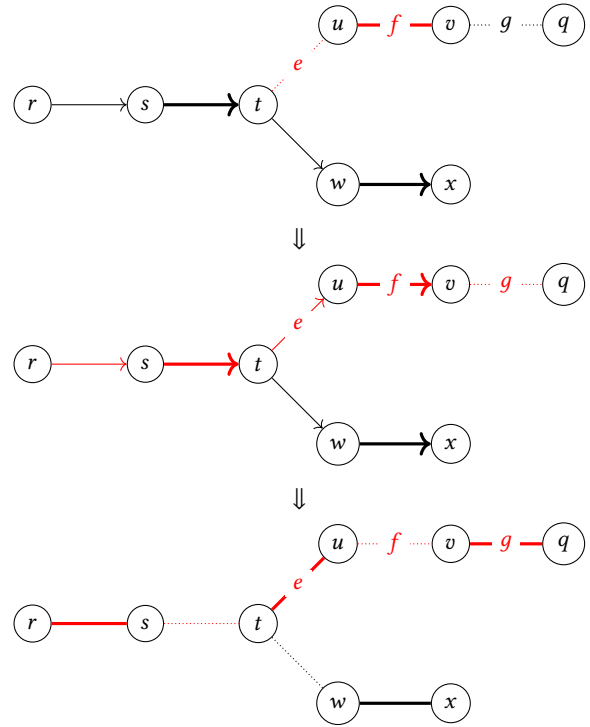

\begin{remark}[{\cite{Kuhn}, \cite[Proposition 5.7]{CCPS}}]\label{BipartiteAlgorithm}
Used together, these operations make up the Hungarian algorithm, \Cref{HungarianAlgorithm}, for producing maximum matchings on bipartite graphs.
\end{remark}

\begin{algorithm}[t]
\SetAlgoLined
\KwData{Bipartite graph $G$, intermediate matching $M$}
\KwResult{Maximum matching on $G$}
%  order the vertices in $G$\;
\While{true}{
    $G_L \cup G_R \leftarrow$ a vertex $2$--coloring of $G$\;
    $T \leftarrow \text{an unmatched vertex in $G_L$}$\;
    \While{true}{
         % augment case
        \uIf{there is an $e = (v, w)$ with $v \in T \cap G_L$, $w \in G_R, \not\in M$}{
            augment $T$ along $e$\;
            break\;
        }
        % graft case
        \uElseIf{there is an $e = (v, w)$ with $v \in T \cap G_L$, $w \in G_R \cap M$}{
            $m \leftarrow \text{match edge for $w$}$\;
            graft $m$ onto $T$ using $e$\;
        }
        \Else{
        \Return{}\;
        }
    }   
 }
 \caption{Hungarian algorithm}
 \label{HungarianAlgorithm}
\end{algorithm}

\begin{example}
We illustrate using an alternating tree to find a maximum matching on a bipartite graph in \Cref{AugmentAndGraftFigure}.
\end{example}

\subsection{Contract and expand blossom}

We now trade the bipartite assumption for two new tree operations.
Consider the situation of \Cref{NonbipartiteFigure}.
The alternating tree $T$ is ``maximally grafted'', but no edge emanating from its positive vertices reaches an unmatched vertex outside of $T$, so the algorithm of \Cref{HungarianAlgorithm} cannot make progress.
Nonetheless, an augmenting path exists: starting at $r$, one can proceed down the lower branch, cross vertically along $e$ to the upper branch, walk backwards through the tree to $u$, and finally cross $f$ to $q$.
These cycles, where an edge not in $T$ joins two of its positive vertices, are the essential new complication of the non-bipartite case.
Edmonds's first fundamental observation was that \emph{all} of the vertices within such a cycle are well-suited to constructing an augmenting path, and the second was that incorporating this into the search algorithm permits one to use it on an arbitrary graph.

\begin{definition}
Let $M$ be an intermediate matching on $G$, and let $v \in G$ be a vertex.
By an \define{alternating cycle rooted at $v$}, we mean an odd-length alternating path $C \subseteq G$ of distinct edges leading from $v$ and returning to $v$.%
\footnote{In particular, $C$ begins and ends with unmatched edges, lest $v$ participate in two matched edges.}
We say that we \define{contract a blossom} from $C$ when we contract (the full subgraph spanned by) $C$ to a point to produce a new graph $G' = G/C$.
We refer to the vertex $B$ in $G'$ which is the image of $C$ as the \define{blossom} or the \define{macrovertex}.%
\footnote{%
Some authors refer instead to the alternating cycle as the blossom.
}
The graph $G'$ inherits a matching $M'$, defined through three cases:
\begin{enumerate}
    \item If $v$ is matched in $M$, $B$ inherits that match in $M'$.
    \item The matched edges in $M$ internal to $C$ are discarded, as they've been contracted out of $G'$.
    \item All other matched edges in $M$ do not interact with $C$, so $M'$ inherits them verbatim.
\end{enumerate}
\end{definition}

\begin{example}
In \Cref{ContractAndAugmentFigure} we illustrate macrovertex contraction.
\end{example}

\begin{remark}
Note that, even if $G$ is singly edged (i.e., is not a multigraph), this need not be the case for the derived graph $G'$ after contracting a cycle $C \subseteq G$ to a macrovertex $B \in G'$.
If there are two distinct vertices $c, c' \in C$ both with edges to a third vertex $d \not\in C$, then $B$ inherits two distinct edges with target $d$.
% We illustrate this phenomenon in \Cref{MultigraphFigure}.
\end{remark}

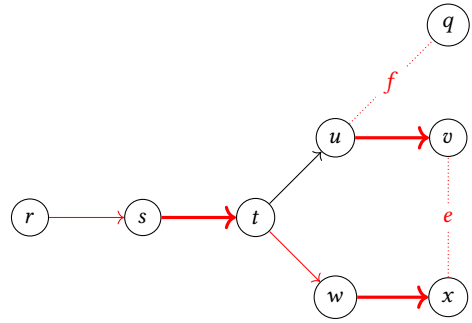
\begin{figure}
\begin{tikzpicture}[main/.style = {draw, circle}, node distance=15mm]
\node[main] (r) {$r$};
\node[main] (s) [right of=r] {$s$};
\node[main] (t) [right of=s] {$t$};

\node[main] (u) [above right of=t] {$u$};
\node[main] (v) [right of=u] {$v$};

\node[main] (w) [below right of=t] {$w$};
\node[main] (x) [right of=w] {$x$};

\node[main] (q) [above of=v] {$q$};

\draw[red, ->] (r) -- (s);
\draw[red, ->, very thick] (s) -- (t);
\draw[->] (t) -- (u);
\draw[red, ->, very thick] (u) -- (v);
\draw[red, ->] (t) -- (w);
\draw[red, ->, very thick] (w) -- (x);

\draw[red, densely dotted] (u) -- (q) node[midway, fill=white] {$f$};
\draw[red, densely dotted] (x) -- (v) node[midway, fill=white] {$e$};

\addtikzpadding;
\end{tikzpicture}
\caption{%
A non-bipartite scenario, where edge $e$ joins two positive vertices.
From the view of \Cref{HungarianAlgorithm}, it is illegal to augment along the edge $f$ because $u$ is not positive, hence the path through $T$ to its root $r$ is not alternating.
However, the edge $e$ could be used to produce an augmenting path, in red.
}
\label{NonbipartiteFigure}
\end{figure}

\begin{definition}
Conversely, suppose that we are given an intermediate matching $N'$ on a graph $G'$, where $G'$ was originally contracted from a graph $G$ along an alternating cycle $C$.
The matching $N'$ can then always be lifted to a matching $N$ on $G$, called the \define{expansion} of $N'$ along $C$.
Namely, note that the macrovertex participates in at most one matched edge in $N'$, which can be identified uniquely with an edge in $G$ incident on some vertex $w \in C$.
By rotating the alternating pattern of matches within $C$ so that the successive pair of unmatched edges is joined at $w$, and otherwise inheriting the matched edges from $N'$, we obtain a matching on $N$.
\end{definition}

\begin{example}
We illustrate match expansion in \Cref{ExpandFigure}.
\end{example}

\begin{figure}
\begin{tikzpicture}[main/.style = {draw, circle}, node distance=14mm]
\node[main] (r) {$r$};
\node[main] (s) [right of=r] {$s$};
\node[main] (t) [right of=s] {$t$};

\node[main] (u) [above right of=t] {$u$};
\node[main] (v) [right of=u] {$v$};

\node[main] (w) [below right of=t] {$w$};
\node[main] (x) [right of=w] {$x$};

\node[main] (q) [above of=v] {$q$};

\draw[->] (r) -- (s);
\draw[->, very thick] (s) -- (t);
\draw (t) to [bend left] (u);
\draw (u) to [bend left] (v);
\draw (t) to [bend right] (w);
\draw (w) to [bend right] (x);

\draw[red, densely dotted] (u) -- (q) node[midway, fill=white] {$f$};
\draw (x) to [bend right] (v);

\node[red, draw, inner sep=2mm, label=above left:$B$, fit=(t) (u) (v) (w) (x)] {};
\addtikzpadding;
\end{tikzpicture}
\[\Downarrow\]
\begin{tikzpicture}[main/.style = {draw, circle}, node distance=14mm]
\node[main] (r) {$r$};
\node[main] (s) [right of=r] {$s$};
\node[main] (t) [right of=s] {$t$};

\node[main] (u) [above right of=t] {$u$};
\node[main] (v) [right of=u] {$v$};

\node[main] (w) [below right of=t] {$w$};
\node[main] (x) [right of=w] {$x$};

\node[main] (q) [above of=v] {$q$};

\draw[red, very thick] (r) -- (s);
\draw[red, densely dotted] (s) -- (t);
\draw (t) to [bend left] (u);
\draw (u) to [bend left] (v);
\draw (t) to [bend right] (w);
\draw (w) to [bend right] (x);

\draw[red, very thick] (u) -- (q) node[midway, fill=white] {$f$};
\draw (x) to [bend right] (v);

\node[draw, inner sep=2mm, label=above left:$B$, fit=(t) (u) (v) (w) (x)] {};
\end{tikzpicture}
\caption{%
Continuing from \Cref{NonbipartiteFigure}, we contract the cycle formed by the edge $e$ into a macrovertex $B$.
This causes $f$ to be attached to a positive vertex in $G'$, and hence it participates in an augmenting path. Note that edges in the contracted cycle are curved solid lines.
}
\label{ContractAndAugmentFigure}
\end{figure}
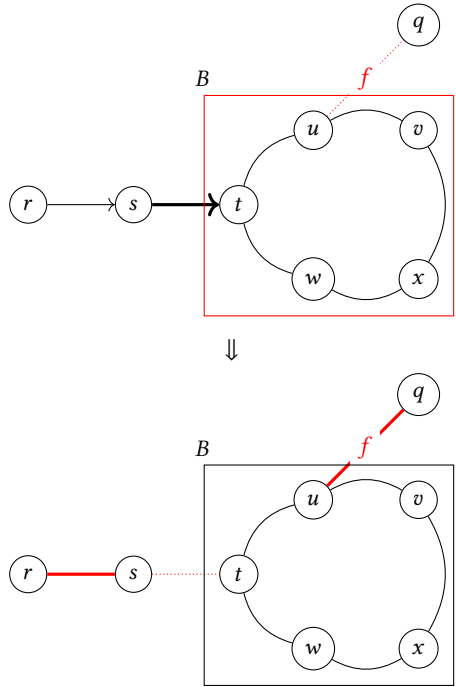

\begin{figure}
\begin{tikzpicture}[main/.style = {draw, circle}, node distance=14mm]
\node[main] (t) {$t$};

\node[main] (u) [above right of=t] {$u$};
\node[main] (v) [right of=u] {$v$};

\node[main] (w) [below right of=t] {$w$};
\node[main] (x) [right of=w] {$x$};

\node[main] (p) [above of=v] {$p$};

\node[main] (r) [below left of=t] {$r$};

\draw[->] (r) -- (w);

\draw (t) to [bend left] (u);
\draw (u) to [bend left] (v);
\draw (t) to [bend right] (w);
\draw (w) to [bend right] (x);
\draw (x) to [bend right] (v);

\draw[->, very thick] (u) -- (p);

\node[red, draw, inner sep=2mm, label=above left:$B$, fit=(t) (u) (v) (w) (x)] {};
\addtikzpadding;
\end{tikzpicture}
\[\Downarrow\]
\begin{tikzpicture}[main/.style = {draw, circle}, node distance=14mm]
\node[main] (t) {$t$};
\node[main] (u) [above right of=t] {$u$};
\node[main] (v) [right of=u] {$v$};
\node[main] (w) [below right of=t] {$w$};
\node[main] (x) [right of=w] {$x$};
\node[main] (p) [above of=v] {$p$};
\node[main] (r) [below left of=t] {$r$};

\draw[->] (r) -- (w);
\draw[->, very thick] (w) -- (t);
\draw[->] (t) -- (u);
\draw[->, very thick] (u) -- (p);
\draw[densely dotted] (u) -- (v);
\draw[densely dotted] (w) -- (x);
\draw[very thick] (x) -- (v);
\end{tikzpicture}
\caption{%
Beginning with a scenario similar to that of \Cref{ContractAndAugmentFigure}, but with $B$ in a negative position, we demonstrate the effect of macrovertex expansion. The edges within the cycle are tagged as matched edges so as to provide an alternating path from the root to the leaf, and edges not along that path are ejected from the tree.
}
\label{ExpandFigure}
\end{figure}
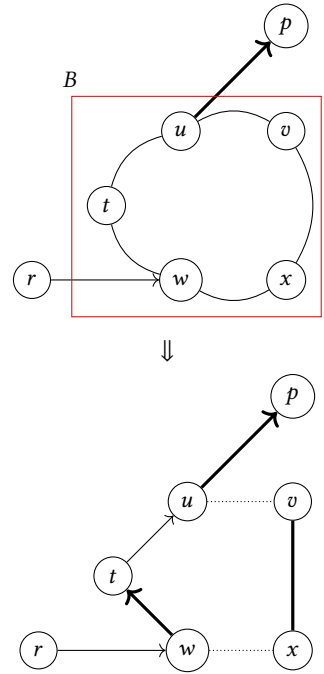

\begin{remark}[{\cite{EdmondsPathsTreesFlowers}, \cite[Theorem 5.10]{CCPS}}]%
\label{UnweightedClassicalAlgorithm}
As promised, we use these operations to extend \Cref{HungarianAlgorithm} to cover non-bipartite graphs.
We also modify the termination condition: given a maximally-grafted alternating tree $T \subseteq G$, if it does not admit any exiting edges along which we may augment, we additionally search for unmatched edges between positive vertices in $T$.
If such an edge is present, then we use it to form a minimum alternating cycle $C$ within $T$, and contract that cycle to form a new graph $G'$ with matching $M'$.
By contracting $T$ along $C$, we also produce a new alternating tree $T'$ in $G'$, and we proceed to run the matching algorithm from this new state.%
\footnote{%
Some authors call this the \define{derived graph} and \define{derived state}.
}
When this recursion returns, we either lift the modified matching on $G'$ to a modified matching on $G$ (via macrovertex expansion) and restart the outer loop, or we proceed to try the next unmatched vertex as a root.
\end{remark}

\begin{remark}
Representing a contracted graph in memory is somewhat onerous.
In particular, the algorithm described in \Cref{UnweightedClassicalAlgorithm} may involved nested contractions, where a vertex becomes contracted into a macrovertex, which in turn becomes contracted into another macrovertex, and so on.
We defer discussion of this point to \Cref{DistributedContractExpandSection}, where we will describe it in full in the setting most relevant to this paper.
\end{remark}

\subsection{Reweight}

Finally, \Cref{UnweightedClassicalAlgorithm} can be extended to produce maximum matchings of minimum weight by housing it in a ``primal-dual update'' scheme~\cite{DFF}.
Speaking very loosely, the dual step sorts the edges not yet considered by the amount of weight their participation would incur, and the primal step consists of \Cref{UnweightedClassicalAlgorithm} as applied to these minimally-sifted graphs.

\begin{definition}
The \define{internal weight} of a (macro)vertex is a numeric value managed by the dual step.
The \define{adjusted weight} of an edge is its weight after subtracting the internal weights of its two endpoints and those of any macrovertices to which they belong.
The \define{weightless subgraph} $G_\circ$ is then the maximal subgraph with the same vertices but retaining only edges of adjusted weight zero.
\end{definition}

\begin{remark}[{\cite[Theorem 5.20]{CCPS}}]\label{AdjEdgeWeightsAreNonneg}
The internal weight of a vertex may be negative.
However, if the edge weights of a graph are all nonnegative, then so are the algorithm's internal weights and adjusted edge weights.
\end{remark}

\begin{definition}\label{ReweightOperation}
Let $G$ be an edge-weighted graph with internal vertex weights, $M$ an intermediate matching on $G_\circ$, and $T$ an alternating tree in $G_\circ$.
The \define{reweighting} of $T$ is an update to the internal weights of $G$ given by increasing the internal weights of the positive vertices $T_+$ and decreasing the internal weights of the negative vertices $T_-$ by the amount given by the minimum of the following three sets:
\begin{description}
    \item[Graft/augment candidates] The adjusted edge weights of edges in $G$ joining vertices in $T_+$ to vertices not in $T$.
    \item[Contract candidates] The adjusted edge weights of edges in $G$, scaled down by half, joining vertices in $T_+$ to each other.
    \item[Expand candidates] The internal weights of vertices in $T_-$ which are macrovertices.
\end{description}
This enlarges $G_\circ$ while maintaining the weightlessness of $T$.
\end{definition}

The edge-weighted blossom algorithm is then given by applying the non-bipartite blossom algorithm to $G_\circ$, with the additional step that the alternating subtree should attempt a reweight operation before the loop gives up and moves on to the next candidate root vertex.

\begin{example}
We illustrate a sample run of this algorithm in \Cref{ReweightFigure}.
\end{example}

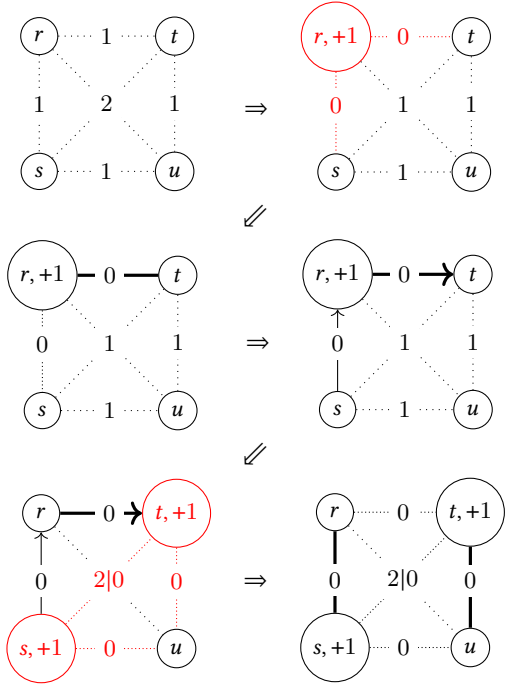
\begin{figure}
\begin{tabular}{lll}
\raisebox{-.42\height}{
\begin{tikzpicture}[main/.style = {draw, circle}, node distance=18mm]
\node[main] (r) {$r$};
\node[main] (s) [below of=r] {$s$};
\node[main] (t) [right of=r] {$t$};
\node[main] (u) [right of=s] {$u$};
\draw[dotted] (r) -- (t) node[midway, fill=white]{$1$};
\draw[dotted] (r) -- (s) node[midway, fill=white]{$1$};
\draw[dotted] (s) -- (u) node[midway, fill=white]{$1$};
\draw[dotted] (t) -- (u) node[midway, fill=white]{$1$};
\draw[dotted] (s) -- (t);
\draw[dotted] (r) -- (u) node[midway, fill=white]{$2$};
\addtikzpadding;
\end{tikzpicture}
}
&
\hspace{0.15cm}
$\Rightarrow$
&
\raisebox{-.39\height}{
\begin{tikzpicture}[main/.style = {draw, circle}, node distance=18mm]
\node[red, main] (r) {$r,+1$};
\node[main] (s) [below of=r] {$s$};
\node[main] (t) [right of=r] {$t$};
\node[main] (u) [right of=s] {$u$};
\draw[red, densely dotted] (r) -- (t) node[pos=0.4, fill=white]{$0$};
\draw[red, densely dotted] (r) -- (s) node[pos=0.41, fill=white]{$0$};
\draw[dotted] (s) -- (u) node[midway, fill=white]{$1$};
\draw[dotted] (t) -- (u) node[midway, fill=white]{$1$};
\draw[dotted] (r) -- (u);
\draw[dotted] (s) -- (t) node[midway, fill=white]{$1$};
\addtikzpadding;
\end{tikzpicture}
}
\end{tabular}
\[\Swarrow\]
\begin{tabular}{lll}
\hspace{-0.23cm}
\raisebox{-.43\height}{
\begin{tikzpicture}[main/.style = {draw, circle}, node distance=18mm]
\node[main] (r) {$r,+1$};
\node[main] (s) [below of=r] {$s$};
\node[main] (t) [right of=r] {$t$};
\node[main] (u) [right of=s] {$u$};
\draw[-, very thick] (r) -- (t) node[pos=0.4, fill=white]{$0$};
\draw[densely dotted] (r) -- (s) node[pos=0.41, fill=white]{$0$};
\draw[dotted] (s) -- (u) node[midway, fill=white]{$1$};
\draw[dotted] (t) -- (u) node[midway, fill=white]{$1$};
\draw[dotted] (r) -- (u);
\draw[dotted] (s) -- (t) node[midway, fill=white]{$1$};
\end{tikzpicture}
}
&
\hspace{0.13cm}
$\Rightarrow$
&
\raisebox{-.425\height}{
\begin{tikzpicture}[main/.style = {draw, circle}, node distance=18mm]
\node[main] (r) {$r,+1$};
\node[main] (s) [below of=r] {$s$};
\node[main] (t) [right of=r] {$t$};
\node[main] (u) [right of=s] {$u$};
\draw[->, very thick] (r) -- (t) node[pos=0.39, fill=white]{$0$};
\draw[->] (s) -- (r) node[pos=0.58, fill=white]{$0$};
\draw[dotted] (s) -- (u) node[midway, fill=white]{$1$};
\draw[dotted] (t) -- (u) node[midway, fill=white]{$1$};
\draw[dotted] (r) -- (u);
\draw[dotted] (s) -- (t) node[midway, fill=white]{$1$};
\end{tikzpicture}
}
\end{tabular}
\[\Swarrow\]
\begin{tabular}{lll}
\raisebox{-.47\height}{
\begin{tikzpicture}[main/.style = {draw, circle}, node distance=18mm]
\node[main] (r) {$r$};
\node[red, main] (s) [below of=r] {$s, +1$};
\node[red, main] (t) [right of=r] {$t, +1$};
\node[main] (u) [right of=s] {$u$};
\draw[->, very thick] (r) -- (t) node[pos=0.6, fill=white]{$0$};
\draw[->] (s) -- (r) node[pos=0.4, fill=white]{$0$};
\draw[red, densely dotted] (s) -- (u) node[pos=0.42, fill=white]{$0$};
\draw[red, densely dotted] (t) -- (u) node[pos=0.42, fill=white]{$0$};
\draw[dotted] (r) -- (u);
\draw[red, densely dotted] (s) -- (t) node[midway, fill=white]{$2 | 0$};
\end{tikzpicture}
}
&
$\Rightarrow$
&
\raisebox{-.465\height}{
\begin{tikzpicture}[main/.style = {draw, circle}, node distance=18mm]
\node[main] (r) {$r$};
\node[main] (s) [below of=r] {$s, +1$};
\node[main] (t) [right of=r] {$t, +1$};
\node[main] (u) [right of=s] {$u$};
\draw[-, densely dotted] (r) -- (t) node[pos=0.6, fill=white]{$0$};
\draw[-, very thick] (s) -- (r) node[pos=0.4, fill=white]{$0$};
\draw[densely dotted] (s) -- (u) node[pos=0.42, fill=white]{$0$};
\draw[-, very thick] (t) -- (u) node[pos=0.42, fill=white]{$0$};
\draw[dotted] (r) -- (u);
\draw[densely dotted] (s) -- (t) node[midway, fill=white]{$2 | 0$};
\end{tikzpicture}
}
\end{tabular}
\caption{%
A demonstration of the weighted blossom algorithm.
Beginning with the root vertex $r$, we first minimally increase its weight so some of its edges become weightless (the densely dotted edges).
We then select one such edge along which to augment the matching.
We choose $s$ to be our next root, to which we graft the previous matched edge.
We are out of weightless edges on which to act, so we minimally reweight the tree to produce more.
We choose the augmenting path $utrs$, along which we augment to produce the final matching.
}
\label{ReweightFigure}
\end{figure}

% note for later: might be nice to show how weights + macrovertices interact in an example

\subsection{The main loop}

Altogether, these operations make up the serial blossom algorithm, \Cref{SerialBlossomAlgorithm}, for producing minimum-weight maximum matchings on edge-weighted graphs.

\begin{remark}
To simplify the outer loop slightly, we have assumed while writing \Cref{SerialBlossomAlgorithm} that $G$ is fully connected with an even number of vertices.
This means that \emph{all} of the vertices will participate in a maximum (perfect) matching, and the objective is only to minimize the total weight.
Matching all of the vertices can be used as a stopping condition.
\end{remark}

% TODO: replace [t] with something more appropriate
\begin{algorithm}[t]
\SetAlgoLined
\KwData{Edge-weighted fully-connected graph $G = (V, E)$ with $|V|$ even}
% \KwData{Edge-weighted, fully-connected, even graph $G$}
\KwResult{Minimum-weight perfect matching on $G$}
%  order the vertices in $G$\;
 weight all the vertices in $G$ to $0$\;
 \While{$G$ has unmatched vertices}{
  $T \leftarrow \text{an unmatched vertex in $G$}$\;
  \While{true}{
   $T_+ \leftarrow \text{vertices of $T$ of even depth}$\;
   $T_- \leftarrow \text{vertices of $T$ of odd depth}$\;
   $G_\circ \leftarrow \text{subgraph of weightless edges}$\;
   % augment case
   \uIf{there is an $e = (v, w) \in G_\circ$ with $v \in T_+$, $w \not\in (M \cup T)$}{
    augment $T$ along $e$\;
    break\;}
   % graft case
   \uElseIf{there is an $e = (v, w) \in G_\circ$ with $v \in T_+$, $w \not\in T$ matched}{
    $m \leftarrow \text{match edge for $w$}$\;
    graft $m$ onto $T$ using $e$\;}
   % blossom case
   \uElseIf{there is an $e = (v, w) \in G_\circ$ with $v, w \in T_+$}{
    $a \leftarrow$ nearest common ancestor of $v$ and $w$\;
    $C \leftarrow$ alternating cycle rooted at $a$\;
    contract $C$ into macrovertex $B$ using $e$\;
    }
   % expand case
   \uElseIf{there is a weightless macrovertex $B \in T_-$}{
    expand $B$\;}
   % reweight case
   \uElseIf{$r \ne 0$ for $r$ the candidate reweight amount for $T$}{
    reweight $T$ by $r$\;
   }
   \Else{
    break\;
    }}}
 % finish expansion
 \While{$G$ contains macrovertices}{
  $B \leftarrow$ a macrovertex in $G$\;
  expand $B$\;}
 \caption{Serial blossom algorithm}
 \label{SerialBlossomAlgorithm}
\end{algorithm}

\begin{theorem}[{\cite{EdmondsMaximumMatching01}, \cite[Theorem 5.16]{CCPS}}]\label{ClassicalCorrect}
\Cref{SerialBlossomAlgorithm} is \emph{correct}: it always terminates, and on termination it emits a perfect matching of minimum weight.
\qed
\end{theorem}

\begin{remark}\label{ResetVariant}
There is a variant of the algorithm in which the main loop resets the tree $T$ after each primal and dual state update, i.e., each update other than grafting.
This variant is less strict than \Cref{SerialBlossomAlgorithm}, in the sense that it admits more execution paths: it can always recreate the state preserved here which it has erased, but \Cref{SerialBlossomAlgorithm} is not always able to erase grafting decisions it has made.
In trade, \Cref{SerialBlossomAlgorithm} performs less recomputation as the state evolves, giving it better runtime properties.
\end{remark}

% \begin{remark}[{\cite[Corollary 26.1a]{Schrijver}}]
% Another popular variant of the outer loop is to reweight all possible vertices at once: if the algorithm is unable to make progress on the weightless subgraph, it calculates sets of all positive and negative nodes (based on their distance to unmatched vertices along alternating paths), reweights all of those nodes according to the rules of \Cref{ReweightOperation}, and reconsiders the new weightless subgraph.
% This variant has the additional pleasant property that it produces a minimum weight matching \emph{of every cardinality} as it proceeds~\cite[Theorem 7]{KalbelWalter}, whereas \Cref{SerialBlossomAlgorithm} produces a minimum weight representative over possible matchings on its particular intermediate set of vertices.
% \end{remark}

\section{The distributed algorithm}

In this section, we rebuild the serial weighted blossom algorithm from \Cref{ClassicalSection} in a concurrent framework.
Our intention is to alleviate a fundamental bottleneck in the serial algorithm: the outer loop fixes an unmatched vertex $v$ to use as the root of an alternating tree $T$, and all operations proceed in view of $T$.
One can imagine a variant of the algorithm which instead constructs a forest of alternating trees, and one can further imagine a decentralized variant where each tree in the forest%
\footnote{%
Indeed, each vertex in the tree.
}
is responsible for ``managing itself''.
This is the manner of algorithm which we now describe.

We adopt the language of the \textsf{DECOUPLED} model of computing~\cite{Linial,CDGFRR,DGFFR}, in which we consider an actor-based programming model~\cite{HBS} which rides atop a message-passing system with guaranteed, ordered delivery of messages.
In a prequel to this paper, we described a specific such system~\cite{PetersonKaralekasAetherPaper} as well as an emulator for it~\cite{aether}, and as a companion to this paper we provide an implementation of our algorithm within that framework~\cite{anatevka}.

\subsection{The distributed environment}\label{DistributedEnvironmentSection}

Our preferred model of distributed computing consists of a family of actors, called \define{processes}, each with a public address, an interrupt table of handlers which service messages arriving at the public address, and a continuation representing a current computational state.
On each computational step, the process walks every message received at its public address using the message handlers, then evaluates its current continuation, which either produces a new continuation or terminates the process.
These components are subject to some basic guarantees:
\begin{itemize}
    \item A non-terminated process will eventually act again.
    \item A message sent between processes will eventually arrive.
    \item Messages sent from the same originating process to the same destination address will arrive in order.
\end{itemize}
In practice, the precise timing of these operations is influenced by many factors (e.g., network pressure), and we assume no guarantees about synchronicity or bounded delay.
% A central organizational principle of any distributed algorithm is locality of state: individual processes are responsible for their own roles in the algorithm, and large-scale coordination among processes inhibits their parallel action, so should be avoided whenever possible.
As is common in distributed graph algorithms, we will spawn a process for each vertex in the problem graph, and their communication will run along edges in the problem graph.
The job of these processes, then, is to coordinate with one another, decide which vertices are matched to which others, and signal when they have finished working.

Rather than embed the problem instance directly into the processes, we instead provide an oracle service, called the \define{dryad},%
\footnote{As in: a being which attends to the health and welfare of flora.}
which responds to API requests with details about the problem graph.
These messages are:
\begin{description}
    \item[\texttt{message-discover}] The dryad replies with a list of vertices which are connected by edges to the querying vertex.
    \item[\texttt{message-sprout}] The sender vertex announces that it has begun to participate in the intermediate matching.
\end{description}

% \begin{remark}
% In our application to quantum error correction~\cite{PetersonKaralekasArtemisPaper}, the dryad is \emph{not} a monolithic process, and we take advantage of this to manage the evolution of an online blossom solver and to cut down on message complexity.
% From the perspective of this paper, however, we will not require any functionality beyond the above messages, and we defer all discussion of these considerable complications to the sequel.
% \end{remark}

\subsection{The blossom main loop}
\label{DistributedBlossomMainLoop}

As in the serial case, the operations available to an alternating tree are: graft an external matched edge, augment through another alternating tree,%
\footnote{%
In the serial case, the other tree is a lone vertex, joined to a positive vertex in the tree by an unmatched edge.
}
form a macrovertex from a cycle, expand a macrovertex, reweight its vertices, and do nothing (which we call ``pass'' or ``hold'').
We will first discuss how an alternating tree selects among these operations, called a \define{scan}, turning to the distributed implementation of these operations only in subsequent sections.

Since this procedure makes use of the alternating tree structure, we will need to understand how this state is maintained by the algorithm.
We store the relevant information in the following slots on each blossom process, where a ``blossom'' may represent a vertex or a family of (macro)vertices contracted into a macrovertex:%
\footnote{%
The convention we use going forward of having ``blossom'' mean either vertex or macrovertex is our own, and differs from the one used in \Cref{ClassicalSection}.
}

\begin{description}
  \item[\texttt{addr}]
  The public address of the blossom.
  Messages sent to this address will be handled by the blossom.%
  \footnote{%
  Unless explicitly stated, when we refer to blossoms we are referring to their addresses, not the processes themselves.
  We use Greek letter and calligraphic variables to distinguish addresses from other data.
  }
  \item[\texttt{id}]
  The internal name of the blossom.
  Blossoms also have at their disposal a function called \texttt{edge-weight} which takes in two \texttt{id}s and returns the weight of the edge between them.
  \item[\texttt{match-edge}]
  The edge connecting this blossom to the blossom to which it is matched, if any.
  \item[\texttt{parent}]
  The optional edge connecting this blossom to its parent in an alternating tree, if it is not the root.
  \item[\texttt{children}]
  A list of edges connecting this blossom to its children in an alternating tree, if any.
  \item[\texttt{positive?}]
  Parity of the distance to the root of the alternating tree: \texttt{true} if even, \texttt{false} if odd.
  \item[\texttt{pistil}]
  If this blossom has been contracted as part of a cycle, this holds the address of its immediate ``parent'' macrovertex.
  \item[\texttt{petals}]
  If this blossom is a macrovertex, this stores the cyclic list of edges which were contracted into it.
  \item[\texttt{internal-weight}]
  The internal weight of the blossom.
  \item[\texttt{pingable}]
  A flag indicating whether the blossom responds to inbound ping requests (see immediately below) or it allows them to queue for later perusal.
  \item[\texttt{paused?}]
  If toggled, continue to respond to messages but take no other actions.
  \item[\texttt{dryad}]
  The address of this blossom's dryad.
\end{description}

Most operations that an alternating tree might perform correspond to the edges emanating from the vertices in the tree.
However, the action to which an edge corresponds depends on the states of \emph{both} of the edge's vertices: the ``target'' vertex may belong to the same tree, a different tree (and at even or odd height), or to no tree at all.
To discern the appropriate action, the (positive) source vertex sends a \texttt{message-ping} to a target vertex with its half of the data, then listens for a \texttt{message-pong} in reply with the calculated operation.
Attached to most calculated operations is the sponsoring edge, which in this context is a data structure with 4 address slots: \texttt{source-blossom}, \texttt{source-vertex}, \texttt{target-vertex}, and \texttt{target-blossom}.
The vertex-suffixed slots are self explanatory; the blossom-suffixed slots return the address of the topmost macrovertex associated with this vertex (if one exists), and otherwise return the vertex address.
In this way edges keep track of the vertices and topmost blossoms associated with each endpoint of the directed edge, which will later make it easier to manage macrovertex contraction and expansion.
The procedure for calculating sponsored actions and edges is described in \Cref{PingPongAlgorithm}, and the information provided by the \texttt{message-ping} includes:
\begin{description}
    \item[\texttt{root}] The root blossom of the source alternating tree.
    \item[\texttt{blossom}] The topmost macrovertex to which the source belongs.
    If the source is not contained within a macrovertex, this is (the address of)
    the source vertex itself.
    \item[\texttt{weight}] The sum of the source  vertex's internal weight and the internal weights of all macrovertices to which it belongs.
    % \item[\texttt{reply-address}] The (ephemeral) address to send the resulting \texttt{message-pong} to.
    \item[\texttt{hold-cluster}] A set of roots that will have future implications on how operations are processed (see \Cref{MultireweightSection}).%
\footnote{%
We call this \texttt{internal-roots} or \texttt{internal-root-set} in our implementation, but have renamed it here to more clearly link it to its intended purpose.}
    \item[\texttt{addr}] The address of the source vertex, used to build the sponsoring edge and send a reply.
    \item[\texttt{id}] The name of the source vertex, used to calculate the underlying edge weight joining it to the recipient.
    % \item[\texttt{internal-roots}] Those roots to exclude from augmentation, e.g., the source root.
\end{description}

\begin{algorithm}[t]
\SetAlgoLined
\KwData{a (target) vertex process $P$, a \texttt{message-ping} $m$}
\KwResult{a \texttt{message-pong} reply containing an action $a$}
 $\rho \leftarrow$ the root of the tree to which $P$ belongs\;
 $\beta \leftarrow$ the topmost macrovertex to which $P$ belongs\;
 $w \leftarrow$ the macrovertex-adjusted  internal weight of $P$\;
 $\mathcal H \leftarrow$ \texttt{hold-cluster($m$)}\;
 $e \leftarrow \texttt{edge(blossom($m$)}, \texttt{addr($m$)}, \texttt{addr($P$)}, \beta\texttt{)}$\;
 $w_e \leftarrow \texttt{edge-weight(id($P$)},\texttt{id($m$))}$\;
 $w_e' \leftarrow w_e - w - \texttt{weight}(m)$\;
 \uIf{$P$ is matched, and $\texttt{match-edge($P$)}$ $= e$}{
  $a \leftarrow$ \texttt{pass}\;
 }
 \uElseIf{$\beta$ $=$ \texttt{blossom($m$)}}{
  $a \leftarrow$ \texttt{pass}\;
 }
 \uElseIf{$\beta$ is negative}{
  $a \leftarrow$ \texttt{hold($\rho$)}\;
 }
 \uElseIf{$w_e'$ is nonzero}{
  \If{$\rho = \texttt{root($m$)}$, or $\rho \in \mathcal H$ and $\texttt{root($m$)}$ $\in \mathcal H$}{
   $w_e' \leftarrow w_e'/2$\;
%   $w' \leftarrow$ $w/2$ \textbf{if} $p = \mathtt{root(\mathit \rho)}$ \textbf{else} $w$\;
  }
  $a \leftarrow$ \texttt{reweight($w_e'$)}\;
%   \uIf{$\rho$ $=$ \texttt{root($m$)}}{$a \leftarrow$ \texttt{reweight($w/2$)}\;}
%   \Else{$a \leftarrow$ \texttt{reweight($w$)}\;}
 }
 \uElseIf{$\rho$ $\ne$ \texttt{root($m$)}, and $\rho$ is matched}{
  $a \leftarrow \texttt{graft($e$)}$\;
 }
 \uElseIf{$\rho$ $\ne$ \texttt{root($m$)}, and $\rho$ is unmatched}{
  % if they're both in the internal root set, halve the weight
  $a \leftarrow \texttt{augment($e$)}$\;
 }
 \ElseIf{$\rho$ $=$ \texttt{root($m$)}}{
  $a \leftarrow \texttt{contract($e$)}$\;
 }
 send \texttt{message-pong($a$)} to \texttt{addr($m$)\;}
\caption{\texttt{message-ping} handler}
\label{PingPongAlgorithm}
\end{algorithm}

\begin{remark}
The slots on a \texttt{message-ping} and on a blossom process are locally accessible during the execution of \Cref{PingPongAlgorithm}.
But, of course, not all data that we use during the course of an algorithm are locally computable.
For example, $P$ might send messages to other processes to compute the addresses $\beta$ and $\rho$ (and to compute properties of their underlying processes, e.g. whether they are matched), in \Cref{PingPongAlgorithm}.
\end{remark}

\begin{algorithm}[t]
\SetAlgoLined
\KwData{a blossom process $B$, a \texttt{message-scan} $m$}
\KwResult{(optional) a supervisor sponsoring an action $a$}
 $\rho  \leftarrow \texttt{root($m$)}$\;
 $\beta  \leftarrow \texttt{blossom($m$)}$ \textbf{if} $\texttt{pistil($B$)}$ \textbf{else} $\texttt{addr($B$)}$\;
 $w  \leftarrow \texttt{weight($m$)} + \texttt{internal-weight($B$)}$\;
 $\mathcal H \leftarrow \texttt{hold-cluster($m$)}$\;
 $\alpha \leftarrow \texttt{addr($B$)}$\;
 $\delta  \leftarrow \texttt{dryad($B$)}$\;
 $E \leftarrow \texttt{children($B$)}$\;

 \uIf{$\texttt{positive?($B$)}$}{
  $a$ $\leftarrow \texttt{pass}$\;
  $E \leftarrow E \cup \texttt{petals($B$)}$\;
 }
 \ElseIf{$B$ is a macrovertex}{
  $a \leftarrow \texttt{expand($\alpha$)}$\;
 }
 
 \If{$\texttt{positive?($B$)}$ and $B$ is not a macrovertex}{
  send \texttt{message-discover} to $\delta$, store response in  $\mathcal N$\;
  \ForEach{$\nu \in \mathcal N$}{
%   $\alpha \leftarrow$ register a new reply address\;
   $p \leftarrow \texttt{message-ping(}\rho, \beta,  w, \mathcal H, \alpha, \texttt{id($B$))}$\;
   send $p$ to $\nu$, store response in $a'$\;
   $a \leftarrow \texttt{unify-pongs($a, a', \mathcal H$)}$\;
   }
 }
 % limit\;  this is a QEC tweak
 \ForEach{$e \in E$}{
  $\tau \leftarrow \texttt{target-blossom($e$)}$\;
%   $\alpha \leftarrow$ register a new reply address\;
  $s \leftarrow \texttt{message-scan(}\rho, \beta,  w, \mathcal H, \alpha$\texttt{)}\;
  send $s$ to $\tau$, store response in $a'$\;
  $a \leftarrow \texttt{unify-pongs($a, a', \mathcal H$)}$\;
 }
%  \If{$\texttt{positive?}(B)$}{
%   $\mathcal N$ $\leftarrow$ (query $\delta$ to discover neighbors of $B$)\;
%   \ForEach{$\nu \in \mathcal N$}{
%   send \texttt{ping} to $\nu$, store response in $a'$\;
%   $a \leftarrow \texttt{unify-pongs}\mathtt{(\mathit a, \mathit{a'})}$
%   }
%   \ForEach{$\gamma \in \mathtt{petals(\mathit B)}$}{
%   send \texttt{scan} to $\gamma$, store response in $a'$\;
%   $a \leftarrow \texttt{unify-pongs}\mathtt{(\mathit a, \mathit{a'})}$
%  }}
%  \Return{a}
 \uIf{\texttt{addr($m$)}}{
  send \texttt{message-pong($a$)} to \texttt{addr($m$)}\;
 }
 \ElseIf{$B$ is still an eligible root, and $a \ne$ \texttt{pass}}{
  pause $B$ and spawn a supervisor sponsoring $a$\;
 }

 \caption{\texttt{message-scan} handler}
 \label{ScanAlgorithm}
\end{algorithm}

We coordinate the different possible pings within a tree using the procedure described in \Cref{ScanAlgorithm}.
This algorithm walks over the edges attached to the vertices participating in the tree, asking each to sponsor an operation.
Because these operations carry a preference order, in the sense that the availability of one operation can preclude the consideration of another, we can then unify the responses into a single course of action according to the following ordered set of rules:

\begin{enumerate}
    \item
    It is possible for an edge to recuse itself from sponsoring an operation, typically from a misguided ping (e.g., if a vertex should ping itself).
    We call this situation a \texttt{pass}.
    Given a choice between a \texttt{pass} and any other option, we will always prefer to act by the other option.
    \item
    A \texttt{hold} operation arises when the ping targets a negative vertex.
    In the serial algorithm there are no interactions between positive and negative vertices, but there are implications in the distributed setting which will require us to process \texttt{hold}s distinctly from \texttt{pass}es (see \Cref{MultireweightSection}).
    For now, we will treat both operations the same.\label{UnifyPongs2}
    % , and it indicates an instruction to take no action, since the serial algorithm has no interactions between positive and negative nodes.
    % However, a future consideration will require us to make a modification to how \texttt{hold}s are processed (see \Cref{MultireweightSection}).
    % such an operation carries with it the identifier of the root of the target tree, and if it belongs to the \define{internal root set}, we always prefer the other operation.
    % \item Unless the other option is a \texttt{reweight}, we prefer any other non-null operation to a \texttt{hold}.
    % \item Given a choice between an \texttt{augment} and any non-null, non-\texttt{augment} operation, we prefer the other operation.%
    % \footnote{%
    % This isn't necessary for correctness, but in practice it seems to accelerate the algorithm.
    % }
    \item
    % Given a choice between two \texttt{hold}s, we \texttt{hold} with the union of the two internal root sets.
    As seen in \Cref{PingPongAlgorithm}, \texttt{hold} operations are initialized with the root $\rho$ of the target vertex.
    The operation stores this root in a slot called \texttt{root-bucket}, which is a set (of length one, to start).
    Given a choice between two \texttt{hold}s, we instead \texttt{hold} with the union of the \texttt{root-bucket}s of each individual \texttt{hold}.
    \item
    Given a choice between a \texttt{reweight} and any operation other than a \texttt{reweight}, we prefer the other operation.
    \item
    Given a choice between two \texttt{reweight} operations, we prefer to \texttt{reweight} by the lesser amount.
    \item
    Finally, given a choice between two nontrivial operations, we arbitrarily prefer one over the other.
\end{enumerate}

\noindent
Altogether, these rules define the method \texttt{unify-pongs} referenced in \Cref{ScanAlgorithm}.

Finally, there are the matters of initiating a \texttt{scan} over a tree and acting on the resulting sponsored operation.
We delegate the responsibility of initiation to the root: each unmatched blossom whose \texttt{paused?} flag is not set to \texttt{true} checks whether it is a root in the forest and, if so, sends itself a \texttt{message-scan}, then awaits a reply.
The slot specification for a \texttt{message-scan} is identical to that of a \texttt{message-ping}, but without an \texttt{id} slot.
When sending the initial \texttt{message-scan} to itself, the root fills out the message slots as follows: \texttt{root} is the root’s address, \texttt{blossom} is null, \texttt{weight} is zero, \texttt{addr} is null, and \texttt{hold-cluster} is a set containing just the root’s address.\label{ScanArguments}
If the result of this \texttt{scan} is not a \texttt{pass}, and if the root has not changed state since the start of the \texttt{scan}, it spawns a \define{supervisor} process seeded with the result and sets its own \texttt{paused?} flag to \texttt{true}.
It is then the supervisor's responsibility to carry out the operation and to unpause the root when complete.

\begin{remark}
In our implementation, we describe the operation sponsored by an edge as a pair of an atom (\texttt{pass}, \texttt{graft}, \texttt{augment}, \texttt{contract}, \texttt{expand}, \texttt{hold}) and a numeric value which records the adjusted weight of the edge.
In particular, we do not use a separate atom for reweighting: a \texttt{reweight} directive is instead recorded by an edge of nonzero weight.
The atom is then determined by the rules in \Cref{PingPongAlgorithm} and \Cref{ScanAlgorithm}.
\end{remark}

\subsection{Augment and graft}

We turn now to modifying the grafting and augmenting operations for distributed use.
These are the simplest operations, which lets us ease into the problem while devoting extra time to the general scaffolding we re-use in future operations.

Recall that each operation is enacted by a supervisor, which is spawned by a root blossom at the conclusion of a \texttt{scan}.
Each such supervisor operation follows the same outline:\label{supervisor-outline}

\begin{enumerate}
    \item Acquire (recursive) locks on the trees involved \cite{PetersonKaralekasAetherPaper}.
    This has the effect of preventing all locked blossoms from starting \texttt{scan}s.
    As part of locking, we additionally set the \texttt{pingable} flag to \texttt{false} to prevent locked blossoms from responding to pings, in order to avoid exposing incomplete state.
    If it is not possible to establish any of the locks, abort.
    \item Check the root blossoms advertised to the supervisor.
    If they are no longer unmatched roots, abort.
    \item 
    % By sending another \texttt{message-ping}, check that the edge responsible for sponsoring the action still sponsors the same action.
    Check that the sponsored action is still valid.
    For most actions (\texttt{graft}, \texttt{augment}, \texttt{contract}), this means sending another \texttt{message-ping} to check that the edge responsible for sponsoring the action still sponsors the same action.
    For \texttt{expand}, this means checking that the sponsoring blossom is still a negative macrovertex.
    We will defer discussion on checking the validity of \texttt{reweight} and \texttt{hold} to later sections.
    If the sponsored action is no longer valid, abort.\label{supervisor-outline-3}
    \item Enter the critical section.
    Send messages instructing the locked vertices to make the appropriate state changes.
    \item Exit the critical section, release the locks, restore ping responsiveness by setting the \texttt{pingable} flag to \texttt{true}, and terminate the supervisor process.
\end{enumerate}

For the \texttt{graft} operation, the supervisor is required to lock the tree from which it came, as well as both blossoms participating in the isolated matched edge to be grafted.%
\footnote{%
In our implementation, we refer to an isolated matched edge as a \define{barbell}.
Each blossom in a barbell considers itself to be the root.
}
The critical section is then constituted of three pairs of set instructions:
\begin{itemize}
  \item Set \texttt{positive?} on both blossoms in the matched edge.
  \item Set \texttt{parent} on both blossoms in the matched edge.
  \item Set \texttt{children} on both blossoms in the (unmatched) edge which sponsored the action.
\end{itemize}

\begin{algorithm}[t]
\SetAlgoLined
\KwData{a blossom process $B$, a \texttt{message-augment} $m$
% a preceding edge $e$,
% a dryad address $\delta$,
% a reply address $\alpha$
}
\KwResult{nothing}
 $e \leftarrow \texttt{preceding-edge($m$)}$\;
 $\delta  \leftarrow \texttt{dryad($B$)}$\;
 \uIf{$e = \texttt{match-edge($B$)}$}{
  $\texttt{match-edge($B$)} \leftarrow \texttt{parent($B$)}$\;
 }
 \Else{
  $\texttt{match-edge($B$)} \leftarrow e$\;
 }
 \uIf{$B$ is the tree root}{
  send a \texttt{message-sprout} to $\delta$\;
  send a ``done'' message to \texttt{addr($m$)}\;
 }
 \Else{
  $e \leftarrow \texttt{reverse(parent($B$))}$\;
  $\pi \leftarrow \texttt{target-blossom(parent($B$))}$\;
  send \texttt{message-augment(addr(}$m$\texttt{)}$,e$\texttt{)} to $\pi$\;
 }
%  \Return{}
\caption{\texttt{message-augment} handler}
\label{AugmentAlgorithm}
\end{algorithm}

The \texttt{augment} operation is more complex.
After locking both the source tree and target tree, the supervisor sends \texttt{message-augment} to the two blossoms in the sponsoring edge.
This message has a slot \texttt{addr} containing the supervisor's address and a slot \texttt{preceding-edge}, initially containing the sponsoring edge.
This causes the blossoms to modify their matches and forward the \texttt{message-augment} toward the roots of their respective trees; details are given in \Cref{AugmentAlgorithm}.\footnote{%
 We introduce here a function \texttt{reverse} to reverse the direction of an edge.
 Note that the supervisor needs to reverse the sponsoring edge before sending it via \texttt{message-augment} to the target tree.}
An additional twist is that, when releasing the recursive lock, we instruct all of the locked blossoms---not just the ones participating in the \texttt{message-augment} chain---to reset their \texttt{parent} and \texttt{children} fields.
This dissembles the tree structures, which are no longer alternating.

\subsection{Contract and expand blossom}\label{DistributedContractExpandSection}

We turn to even more complicated tree operations: first, the construction of macrovertices; and later, their dissolution.

As indicated in the previous section, macrovertex construction follows the same general template of locking and re-checking before entering the critical section, whose behavior we now specify.
To begin, the supervisor spawns a new blossom process which will serve as the macrovertex and immediately acquires a lock on it to prevent it from acting on its own (viz., initiating a \texttt{scan}).
It then takes the following steps:
\begin{description}
  \item[Petal calculation]
  The supervisor computes the paths in the alternating tree from the sponsoring blossoms to the root, trims edges which are common to both paths, and joins them through the sponsoring edge.
  This is the minimal alternating cycle containing the sponsoring edge, and is used to set the \texttt{petals} slot on the new macrovertex.%
  \footnote{%
  In \Cref{NonbipartiteFigure}, the edge $e = (v, x)$ would sponsor a macrovertex formation, and the two paths from the root would read as $rstuv$ and $rstwx$.
  After trimming the common prefix, one resulting alternating cycle is $tuvxw$.
  } The blossoms in the cycle set their \texttt{pistil} to the macrovertex.
  \item[Parent calculation]
  The final such edge trimmed above, if any, connects to the parent of the macrovertex.
  The supervisor substitutes the new macrovertex into that parent's \texttt{children}, and also sets the \texttt{parent} of the macrovertex.
  \item[Children calculation]
  Where appropriate, the blossoms in the cycle transfer the parent-child relationship of their children to the macrovertex.
  \item[Match calculation]
  The macrovertex inherits the match edge of its contracted blossom which is nearest to the root, if it is matched.
  The blossoms in the cycle erase their match edges.
  \item[Unpause]
  Finally, the macrovertex is unpaused.
\end{description}

\begin{remark}
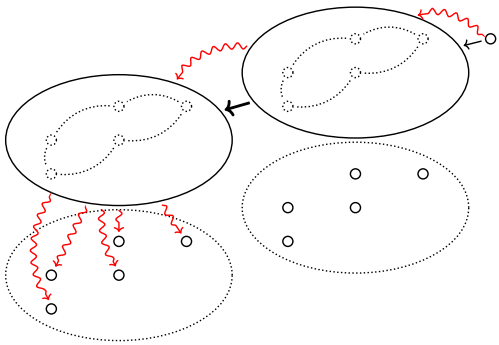
\begin{figure}
\vspace{0.2cm}
\resizebox{0.4\textwidth}{!}{
\begin{tikzpicture}[
  very thick,
  every node/.style={draw,circle},
  % mondenode/.style={fill=black!10!yellow},
  % archnode/.style={fill=black!30!cyan},
  % circletnode/.style={fill=black!30!orange},
  every fit/.style={ellipse,draw},
  shorten >= 3pt,shorten <= 3pt,
%   thin,
]

\begin{scope}[node distance=20mm]
% actual root
\node[color=white] (actualroot) {};
% cycle1
\node[left of=actualroot, yshift=-10mm] (cycle1arm2lower) {};
\node[left of=cycle1arm2lower] (cycle1arm2upper) {};
\node[left of=cycle1arm2lower, yshift=-10mm] (cycle1arm1lower) {};
\node[left of=cycle1arm1lower, yshift=-10mm] (cycle1arm1upper) {};
\node[left of=cycle1arm2upper, yshift=-10mm] (cycle1root) {};
% cycle2
\node[left of=cycle1root, yshift=-10mm, node distance=30mm] (cycle2arm2lower) {};
\node[left of=cycle2arm2lower] (cycle2arm2upper) {};
\node[left of=cycle2arm2lower, yshift=-10mm] (cycle2arm1lower) {};
\node[left of=cycle2arm1lower, yshift=-10mm] (cycle2arm1upper) {};
\node[left of=cycle2arm2upper, yshift=-10mm] (cycle2root) {};
% derived graph
\node[above of=actualroot, yshift=1cm] (root) {};
\node[draw, inner sep=2mm, fit=(cycle1arm2lower) (cycle1arm2upper) (cycle1root) (cycle1arm1upper) (cycle1arm1lower), dotted] (actualblossom1) {};
\node[draw, inner sep=2mm, fit=(cycle2arm2lower) (cycle2arm2upper) (cycle2root) (cycle2arm1upper) (cycle2arm1lower), dotted] (actualblossom2) {};
\node[draw, inner sep=2mm, fit=(cycle1arm2lower) (cycle1arm2upper) (cycle1root) (cycle1arm1upper) (cycle1arm1lower), yshift=4cm] (blossom1) {};
\node[draw, inner sep=2mm, fit=(cycle2arm2lower) (cycle2arm2upper) (cycle2root) (cycle2arm1upper) (cycle2arm1lower), yshift=4cm] (blossom2) {};
% derived cycles
\node[dotted, yshift=4cm, left of=actualroot, yshift=-10mm] (cycle1arm2lowerderived) {};
\node[dotted, yshift=4cm, left of=cycle1arm2lower] (cycle1arm2upperderived) {};
\node[dotted, yshift=4cm, left of=cycle1arm2lower, yshift=-10mm] (cycle1arm1lowerderived) {};
\node[dotted, yshift=4cm, left of=cycle1arm1lower, yshift=-10mm] (cycle1arm1upperderived) {};
\node[dotted, yshift=4cm, left of=cycle1arm2upper, yshift=-10mm] (cycle1rootderived) {};
\node[dotted, yshift=4cm, left of=cycle1root, yshift=-10mm, node distance=30mm] (cycle2arm2lowerderived) {};
\node[dotted, yshift=4cm, left of=cycle2arm2lower] (cycle2arm2upperderived) {};
\node[dotted, yshift=4cm, left of=cycle2arm2lower, yshift=-10mm] (cycle2arm1lowerderived) {};
\node[dotted, yshift=4cm, left of=cycle2arm1lower, yshift=-10mm] (cycle2arm1upperderived) {};
\node[dotted, yshift=4cm, left of=cycle2arm2upper, yshift=-10mm] (cycle2rootderived) {};
% cycle edges
\draw[dotted, bend right] (cycle1arm2lowerderived) to (cycle1arm2upperderived);
\draw[dotted, bend right] (cycle1arm2upperderived) to (cycle1rootderived);
\draw[dotted, bend right] (cycle1rootderived) to (cycle1arm1upperderived);
\draw[dotted, bend right] (cycle1arm1upperderived) to (cycle1arm1lowerderived);
\draw[dotted, bend right] (cycle1arm1lowerderived) to (cycle1arm2lowerderived);
\draw[dotted, bend right] (cycle2arm2lowerderived) to (cycle2arm2upperderived);
\draw[dotted, bend right] (cycle2arm2upperderived) to (cycle2rootderived);
\draw[dotted, bend right] (cycle2rootderived) to (cycle2arm1upperderived);
\draw[dotted, bend right] (cycle2arm1upperderived) to (cycle2arm1lowerderived);
\draw[dotted, bend right] (cycle2arm1lowerderived) to (cycle2arm2lowerderived);
% edges
\draw[->] (root) to (blossom1);
\draw[->, line width=2.4pt] (blossom1) to (blossom2);
% messages
\draw[red, ->, decorate, decoration={snake, pre length=3pt, post length=7pt}, bend right=30] (root) to (blossom1);
\draw[red, ->, decorate, decoration={snake, pre length=3pt, post length=7pt}, bend right=30] (blossom1) to (blossom2);
% \draw[->, decorate, decoration={snake, pre length=3pt, post length=7pt}] (blossom1) -- (cycle1root);
% \draw[->, decorate, decoration={snake, pre length=3pt, post length=7pt}, bend right] (blossom1) to (cycle1arm1upper);
% \draw[->, decorate, decoration={snake, pre length=3pt, post length=7pt}, bend right=15] (blossom1) to (cycle1arm1lower);
% \draw[->, decorate, decoration={snake, pre length=3pt, post length=7pt}] (blossom1) to (cycle1arm2upper);
% \draw[->, decorate, decoration={snake, pre length=3pt, post length=7pt}] (blossom1) -- (cycle1arm2lower);
\draw[red, ->, decorate, decoration={snake, pre length=3pt, post length=7pt}] (blossom2) -- (cycle2root);
\draw[red, ->, decorate, decoration={snake, pre length=3pt, post length=7pt}, bend right] (blossom2) to (cycle2arm1upper);
\draw[red, ->, decorate, decoration={snake, pre length=3pt, post length=7pt}, bend right=15] (blossom2) to (cycle2arm1lower);
\draw[red, ->, decorate, decoration={snake, pre length=3pt, post length=7pt}] (blossom2) to (cycle2arm2upper);
\draw[red, ->, decorate, decoration={snake, pre length=3pt, post length=7pt}] (blossom2) -- (cycle2arm2lower);
\end{scope}
\addtikzpadding;
\end{tikzpicture}
} % resizebox
\caption{%
Illustration of the propagation of \texttt{scan} messages, denoted by wavy red arrows, along an alternating tree with macrovertices.
Scan messages are forwarded first among blossoms along the edges that constitute the alternating tree, and then each positive macrovertex which receives a message propagates it to the blossoms which make up the cycle which it encloses.
}
\label{WarOfTheWorlds}
\end{figure}
In this way, alternating trees with macrovertices become ``two-dimensional trees'': there is a parent-child relation which captures potential alternating paths in the contracted graph, as well as a parent-child relation which captures the subsumption of blossoms by macrovertices.
We illustrate this structure in \Cref{WarOfTheWorlds}.
\end{remark}

Expansion is performed ``in reverse'': all of the relations that went into macrovertex formation are recovered whenever they can be, and they are reverted to a neutral state otherwise.
% As a sanity check in this direction, a macrovertex additionally checks that it has no internal weight, it is not contained within a further contracted cycle, and it resides in a negative position within the alternating tree.
The amounts to the following operations:
\begin{description}
    \item[External parent calculation]
    The parent of the macrovertex holds an edge to the macrovertex in its list of children.
    The supervisor ``unwraps'' the target of that edge by one layer: it replaces the edge by an edge with \texttt{target-blossom} set to the penultimate blossom ancestor of its \texttt{target-vertex}.
    It does this by asking \texttt{target-vertex} to recursively message its \texttt{pistil} to find the process whose \texttt{pistil} is equal to \texttt{target-blossom} (which, if \texttt{target-vertex} has only been part of one macrovertex contraction, would be itself).\footnote{%
By recursively message we mean send a message to its \texttt{pistil}, and then have that process send a message to its \texttt{pistil}, and so on.
 }

    \item[External match calculation]
    The match of the macrovertex holds a pointer to the macrovertex in its match edge.
    In the same fashion as external parent calculation, the supervisor descends that edge by one layer of macrovertex contraction.
    \item[Internal match calculation]
    The match edges within the cycle are set between neighbors so that the unique cycle blossom receiving an external match edge does not participate in a match within the cycle.
    \item[Internal parent and children calculation]
    The blossom algorithm makes two guarantees when expanding macrovertices: a macrovertex is always matched; and it either does not participate in an alternating tree (viz., when expanding a blossom at the termination of the algorithm) or, if it does, it lies in a negative position with zero internal weight (viz., when expanding a blossom during the bulk of the algorithm).
    In the second case, it is thus guaranteed to have at most one child (viz., its match) in addition to having a single parent.
    Accordingly, by choosing between clockwise and counterclockwise traversal, the supervisor can find a unique path through the petals which continues the alternating pattern from the ambient tree.
    These edges are assigned parent-child relationships, and all edges not along the path are ejected from the tree.
\end{description}

\subsection{Reweight}

Finally, we consider the most complex of the serial algorithm's tree operations.
Although reweighting only involves setting values internal to a tree, the \emph{validity} is influenced by an indeterminate number of trees.
It is not feasible to ensure the validity of a reweighting operation by acquiring locks, since pessimistically we would have to lock the entire platform.
Instead, we perform a tentative change, then undo (or ``rewind'') if that change is interrupted or seen to be invalid.%
\footnote{%
The simplest form of validation is to apply \Cref{AdjEdgeWeightsAreNonneg} and check for negatively-weighted edges.
}

The possibility of an invalid reweighting is a wrinkle unique to the distributed setting: a change can only become invalid if two trees, which are mutually influencing one another's \texttt{scan}s, both elect to reweight at the same time.
They can also only detect this invalidity by re-scanning---which, without intervention, would result in deadlock, as both trees refuse to reply to inbound pings in their respective critical sections.
In light of this, we introduce a new kind of ping and new state of ping responsiveness, called \define{soft pings}, which are used only to establish weight validity and not to sponsor new operations.

\begin{description}
    \item[\texttt{pingable}]
    This existing flag is extended from a simple on/off switch to three values: \texttt{all}, meaning the blossom responds to all pings (replaces \texttt{true}); \texttt{none}, meaning the blossom defers all ping responses (replaces \texttt{false}); and \texttt{soft}, meaning the blossom only handles soft pings and defers others.

\end{description}

With this new pingability mode in place, we can define how a supervisor checks the validity of a \texttt{reweight} (which we deferred during step 3 in the supervisor outline in Section~\hyperref[supervisor-outline-3]{3.3}).
First, the supervisor tells the source and target roots to set themselves and their trees to only respond to soft pings.%
\footnote{%
We additionally modify the validity checks for \texttt{graft}, \texttt{augment}, \texttt{contract} to do the same before sending their verification ping.
}
Then, the supervisor instructs the source root to initiate a \define{soft scan} (a \texttt{scan} that generates soft pings).
If the resulting recommendation differs from what was originally sponsored, the supervisor aborts.

This change is sufficient to maintain consistent state, but it can still result in livelock~\cite{ASHCROFT1975110}: two trees competing for the same reweighting operation in perfect synchrony can forever tentatively reweight, regret, rewind, and repeat.
To avoid this scenario, we use a priority scheme to break the symmetry: if a tree in the process of reweighting receives a (soft) ping from a higher-priority source which evidences that it has changed its weight by too much, it aborts its own reweight in order to make room for its superior to perform the operation instead.%
\footnote{%
It is possible, but not necessary, to employ this symmetry-breaking mechanism with the locks acquired by other operations.
}

With all this in mind, we describe steps involved in the critical section of the reweighting operation:

\begin{enumerate}
    \item Set the source tree's pingability to \texttt{none}.
    \item Modify the internal weights of the top-level blossoms in the source tree: the positive blossoms are increased by the desired amount, and negative blossoms are decreased.
    \item Set the source tree's pingability to \texttt{soft}.
    \item Tell the source root to perform a soft scan to determine the minimum edge weight emanating from this tree to another.
    This message contains the same arguments as described in Section \hyperref[ScanArguments]{3.2}, except \texttt{addr} is the supervisor's address.\label{ReweightSoftScan}
    \item If this minimum edge weight is negative, rewind.%
    \footnote{%
    This is the point at which one could only rewind halfway and check again.
    In the priority-preferenced scheme, only a higher-priority vertex or a vertex which has finalized its critical section can emit a negative-weight reply, and in both these cases we must rewind.
    }
\end{enumerate}

\begin{example}
See \Cref{ReweightConflictFigure} for an illustration of this resolution.
\end{example}

\begin{remark}
Rather than preventing livelock through priorities, an alternative strategy when two trees conflict through simultaneous reweighting is to have each tree rewind by half of the overshoot,%
\footnote{%
The factor of 1/2 comes from each 1 edge being attached to 2 vertices.
}
check again, and rewind the rest of the way if there is still a problem.
This is an imperfect form of resolution, but it nonetheless seems to be useful in practice.
See \Cref{ReweightConflictBySharingFigure} for an example of this conflict resolution.
\end{remark}

\begin{figure}
\begin{tikzpicture}[main/.style = {draw, circle}, node distance=20mm]
\node[main] (r) {$r$};
\node[main] (s) [right of=r] {$s$};
\node[main] (t) [right of=s] {$t$};
\node[main] (q) [right of=t] {$q$};

\draw[dotted] (r) -- (s) node[midway, fill=white]{$1$};
\draw[dotted] (s) -- (t) node[midway, fill=white]{$2$};
\draw[dotted] (t) -- (q) node[midway, fill=white]{$3$};
\addtikzpadding;
\end{tikzpicture}
\[\Downarrow\]
\begin{tikzpicture}[main/.style = {draw, circle}, node distance=20mm]
\node[main] (r) {$r$};
\node[red, main] (s) [right of=r] {$s, +1$};
\node[red, main] (t) [right of=s] {$t, +2$};
\node[main] (q) [right of=t] {$q$};

\draw[densely dotted] (r) -- (s) node[pos=0.58, fill=white]{$0$};
\draw[red, dotted] (s) -- (t) node[midway, fill=white]{$-1$};
\draw[dotted] (t) -- (q) node[pos=0.42, fill=white]{$1$};
\end{tikzpicture}
\[\Downarrow\]
\begin{tikzpicture}[main/.style = {draw, circle}, node distance=20mm]
\node[main] (r) {$r$};
\node[main] (s) [right of=r] {$s, +1$};
\node[red, main] (t) [right of=s] {$t$};
\node[main] (q) [right of=t] {$q$};

\draw[densely dotted] (r) -- (s) node[pos=0.58, fill=white]{$0$};
\draw[dotted] (s) -- (t) node[pos=0.44, fill=white]{$1$};
\draw[dotted] (t) -- (q) node[midway, fill=white]{$3$};
\end{tikzpicture}
\caption{%
A demonstration of a reweight conflict.
Two vertices $s$ and $t$ both attempt to reweight according to their local understanding of the largest value by which they can reweight.
Before exiting the critical section, they both notice that the edge connecting them has negative adjusted weight, which is illegal.
Taking $t$ to have lower priority than $s$, $t$ reverses its reweight whereas $s$ retains it.
This repairs the violation, and the algorithm has made progress by making the edge $(r,s)$ weightless.
}
\label{ReweightConflictFigure}
\end{figure}
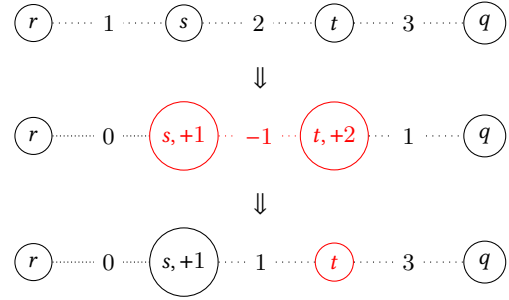

\begin{figure}
% \begin{tikzpicture}[main/.style = {draw, circle}, node distance=20mm]
% \node[main] (r) {$r$};
% \node[main] (s) [right of=r] {$s$};
% \node[main] (t) [right of=s] {$t$};
% \node[main] (q) [right of=t] {$q$};

% \draw[dotted] (r) -- (s) node[midway, fill=white]{$1$};
% \draw[dotted] (s) -- (t) node[midway, fill=white]{$2$};
% \draw[dotted] (t) -- (q) node[midway, fill=white]{$3$};
% \addtikzpadding;
% \end{tikzpicture}
% \[\Downarrow\]
% \begin{tikzpicture}[main/.style = {draw, circle}, node distance=20mm]
% \node[main] (r) {$r$};
% \node[red, main] (s) [right of=r] {$s, +1$};
% \node[red, main] (t) [right of=s] {$t, +2$};
% \node[main] (q) [right of=t] {$q$};

% \draw[densely dotted] (r) -- (s) node[midway, fill=white]{$0$};
% \draw[red, dotted] (s) -- (t) node[midway, fill=white]{$-1$};
% \draw[dotted] (t) -- (q) node[midway, fill=white]{$1$};
% \end{tikzpicture}
% \[\Downarrow\]
\begin{tikzpicture}[main/.style = {draw, circle}, node distance=20mm]
\node[main] (r) {$r$};
\node[red, main] (s) [right of=r] {$s, +1/2$};
\node[red, main] (t) [right of=s] {$t, +3/2$};
\node[main] (q) [right of=t] {$q$};

\draw[dotted] (r) -- (s) node[midway, fill=white]{$1/2$};
\draw[densely dotted] (s) -- (t) node[midway, fill=white]{$0$};
\draw[dotted] (t) -- (q) node[midway, fill=white]{$3/2$};
\end{tikzpicture}
\caption{%
Resolving the conflict in the second step in \Cref{ReweightConflictFigure} by partial rewinding instead.
This time the resulting weightless edge is $(s,t)$.
}
\label{ReweightConflictBySharingFigure}
\end{figure}
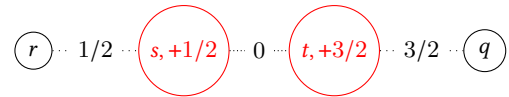

\subsection{Multireweight}\label{MultireweightSection}

Finally, we come to the crucial new feature that arises in the distributed algorithm.
The algorithm described so far suffers from an additional deadlock which arises not because of communication conflicts, but because of decentralized action.
Where the serial algorithm could exhaust the possibilities for a tree to act, \emph{release the active tree}, and move on to the next possible root, the decentralized algorithm builds multiple trees at once and \emph{has no release mechanism}.
Since the \texttt{hold} operation previously indicated that the serial algorithm should move on, this results in a possible failure mode when handling it in the distributed setting.

\begin{example}\label{AbutExample}
For an example configuration, see the second step in \Cref{AbutmentDeadlockFigure}.
These two trees are arranged so that they both have some positive weight on their negative vertices, and they both have positive vertices with weightless edges to each others' negative vertices.
This means that the trees abut along edges that sponsor \texttt{hold} operations, and neither one can perform any other operation, including a reweight.

However, the two trees \emph{could} make progress if they coordinated their action to simultaneously transfer weight from their negative vertices to their positive vertices.
This kind of coordinated reweighting is what is pictured in the transition to the third step of \Cref{AbutmentDeadlockFigure}.
\end{example}

To codify this kind of coordination, we introduce a variant of the reweighting operation, called a \define{multireweight}, which is triggered by the presence of a hold cluster:

\begin{definition}
A \define{hold cluster} is a set of trees whose \texttt{scan}s all elect to \texttt{hold} and all of whose edges sponsoring those \texttt{hold}s have both endpoints in the set.%
\footnote{%
In particular, there are no valid hold clusters of size $1$.
}
\end{definition}

\noindent
% A supervisor responding to a sponsored \texttt{hold} operation then checks whether its tree is actually part of a hold cluster.
% It does this by incrementally growing a set of trees along edges which sponsor \texttt{hold} operations.
% If no new trees are proposed, then this set is a minimal valid hold cluster; if a new tree is proposed that sponsors a non-\texttt{hold} operation, then the cluster is invalid and the supervisor aborts.
A supervisor responding to a sponsored \texttt{hold} operation deviates from the typical recipe described in \Cref{supervisor-outline}.
Namely, it performs the preflight validation checks for the operation (step 3) \textit{before} locking and checking the roots (steps 1 and 2).
For a \texttt{hold}, this validation check consists of verifying that the tree is actually part of a nontrivial hold cluster.
It does this by messaging all the roots in the sponsoring operation’s \texttt{root-bucket} to ask if they are held (and by whom), and aggregates the responses.
If new mutually-held trees (beyond the \texttt{root-bucket}) are encountered, the message is forwarded along to them.
This has the effect of incrementally growing a set of trees along edges which sponsor \texttt{hold} operations.
Once no new trees are proposed, then this set is a valid hold cluster; however, if a new tree is proposed that sponsors a non-\texttt{hold} operation, then the cluster is invalid and the supervisor aborts.

\begin{figure}
% \begin{tikzpicture}[main/.style = {draw, circle}, node distance=15mm]
% \node[main] (r) {$r$};
% \node[main] (u) [right of=r] {$u$};
% \node[main] (v) [right of=u] {$v$};
% \node[main] (s) [below of=r] {$s$};
% \node[main] (t) [left of=s] {$t$};
% \node[main] (q) [right of=s] {$q$};
% \draw[dotted] (r) -- (u) node[midway, fill=white]{$1$};
% \draw[dotted] (r) -- (s) node[midway, fill=white]{$1$};
% \draw[dotted] (s) -- (q) node[midway, fill=white]{$1$};
% \draw[dotted] (u) -- (q) node[midway, fill=white]{$1$};
% \draw[dotted] (t) -- (s) node[midway, fill=white]{$1$};
% \draw[dotted] (u) -- (v) node[midway, fill=white]{$1$};
% \draw[dotted] (s) -- (u);
% \draw[dotted] (r) -- (q) node[midway, fill=white]{$2$};
% \draw[dotted] (r) -- (t) node[midway, fill=white]{$2$};
% \draw[dotted] (v) -- (q) node[midway, fill=white]{$2$};
% \addtikzpadding;
% \end{tikzpicture}
% \[\Downarrow\]
\begin{tikzpicture}[main/.style = {draw, circle}, node distance=14mm]
\node[main] (r) {$r$};
\node[main] (u) [right of=r] {$u, +1$};
\node[main] (v) [right of=u] {$v$};
\node[main] (s) [below of=r] {$s, +1$};
\node[main] (t) [left of=s] {$t$};
\node[main] (q) [right of=s] {$q$};
\draw[densely dotted] (r) -- (u);
\draw[densely dotted] (r) -- (s);
\draw[densely dotted] (s) -- (q);
\draw[densely dotted] (u) -- (q);
\draw[very thick] (t) -- (s);
\draw[very thick] (u) -- (v);
\draw[densely dotted] (s) -- (u);
\draw[dotted] (r) -- (q) node[midway, fill=white]{$2$};
\draw[dotted] (r) -- (t) node[midway, fill=white]{$2$};
\draw[dotted] (v) -- (q) node[midway, fill=white]{$2$};
\addtikzpadding;
\end{tikzpicture}
\[\Downarrow\]
\begin{tikzpicture}[main/.style = {draw, circle}, node distance=14mm]
\node[main] (r) {$r$};
\node[main] (u) [right of=r] {$u, +1$};
\node[main] (v) [right of=u] {$v$};
\node[main] (s) [below of=r] {$s, +1$};
\node[main] (t) [left of=s] {$t$};
\node[main] (q) [right of=s] {$q$};
\draw[densely dotted] (r) -- (u);
\draw[->] (r) -- (s);
\draw[densely dotted] (s) -- (q);
\draw[->] (q) -- (u);
\draw[->, very thick] (s) -- (t);
\draw[->, very thick] (u) -- (v);
\draw[densely dotted] (s) -- (u);
\draw[dotted] (r) -- (q) node[midway, fill=white]{$2$};
\draw[dotted] (r) -- (t) node[midway, fill=white]{$2$};
\draw[dotted] (v) -- (q) node[midway, fill=white]{$2$};
\end{tikzpicture}
\[\Downarrow\]
\begin{tikzpicture}[main/.style = {draw, circle}, node distance=14mm]
\node[red, main] (r) {$r, +1$};
\node[main] (u) [right of=r] {$u$};
\node[red, main] (v) [right of=u] {$v, +1$};
\node[main] (s) [below of=r] {$s$};
\node[red, main] (t) [left of=s] {$t, +1$};
\node[red, main] (q) [right of=s] {$q, +1$};
\draw[densely dotted] (r) -- (u);
\draw[->] (r) -- (s);
\draw[densely dotted] (s) -- (q);
\draw[->] (q) -- (u);
\draw[->, very thick] (s) -- (t);
\draw[->, very thick] (u) -- (v);
\draw[dotted] (s) -- (u);
\draw[red, densely dotted] (r) -- (q) node[midway, fill=white]{$2|0$};

\draw[densely dotted] (v) -- (q);
\draw[densely dotted] (r) -- (t);
\end{tikzpicture}
\[\Downarrow\]
\begin{tikzpicture}[main/.style = {draw, circle}, node distance=14mm]
\node[main] (r) {$r, +1$};
\node[main] (u) [right of=r] {$u$};
\node[main] (v) [right of=u] {$v, +1$};
\node[main] (s) [below of=r] {$s$};
\node[main] (t) [left of=s] {$t, +1$};
\node[main] (q) [right of=s] {$q, +1$};
\draw[densely dotted] (r) -- (u);
\draw[densely dotted] (r) -- (s);
\draw[densely dotted] (s) -- (q);
\draw[densely dotted] (q) -- (u);
\draw[very thick] (t) -- (s);
\draw[very thick] (u) -- (v);
\draw[very thick] (r) -- (q);
\draw[dotted] (s) -- (u) node[midway, fill=white]{$2$};
\draw[densely dotted] (v) -- (q);
\draw[densely dotted] (r) -- (t);
\end{tikzpicture}
\caption{%
An example of a multireweight operation.
In this graph, the edge weights are determined by Manhattan distance.
% First, $s$ and $u$ act as roots, where they each reweight by $1$ and then respectively augment through $t$ and $v$ to enlarge the matching.
First, $r$ and $q$ act as roots, and they each graft the indicated matched edges onto their respective trees.
This causes a serial impasse: the only remaining weightless edge connects two negative vertices, and neither tree can reweight alone because of the weightless edge connecting its root to the other tree's negative vertex.
However, the two trees can together perform a multireweight by $1$, which redistributes the internal weight off of the negative vertices and onto the positive ones, resulting in a weightless edge connecting the two roots which can be augmented to produce a perfect matching.
% Finally, the algorithm augments through that new weightless edge to produce the desired perfect matching.
}
\label{AbutmentDeadlockFigure}
\end{figure}
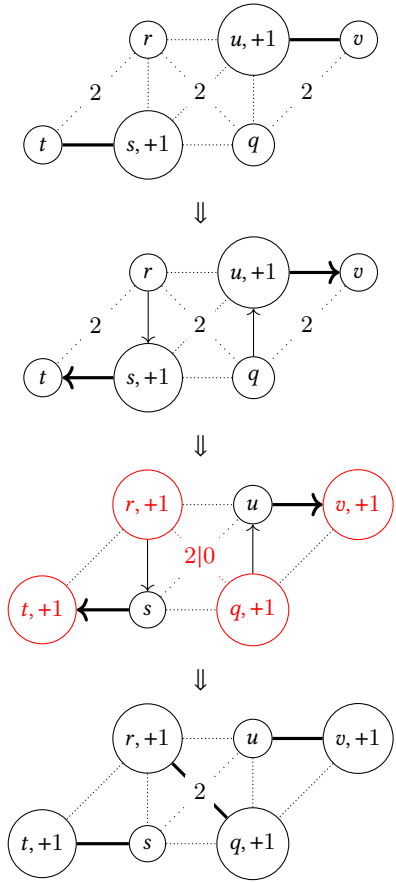

Once the supervisor has established a valid hold cluster, and if it was spawned by the root with the highest priority among the cluster, it proceeds through steps 1 and 2 of the standard recipe.
This locks all the trees in the hold cluster and then checks that their roots have not changed since the spawn of the supervisor, otherwise aborting.
Then, skipping step 3 (as we performed validation already), the supervisor reaches the \texttt{multireweight} critical section.
Before describing the way in which a supervisor enacts a \texttt{multireweight} operation, we must first (as promised in Section \hyperref[UnifyPongs2]{3.2}) make a modification to how \texttt{hold}s are processed by \texttt{unify-pongs}.
Rather than always treating it like a \texttt{pass} and preferring the other operation, \texttt{unify-pongs} only discards a \texttt{hold} when its root is part of the hold cluster, and otherwise considers it a nontrivial operation.

With this modification in place, we can finally describe the \texttt{multireweight} critical section:

\begin{enumerate}
    \item
    Set the pingability of the hold cluster to \texttt{soft}.
    \item
    Perform a soft scan across the entire hold cluster.
    For each root in the cluster, this soft scan resembles that of the \texttt{reweight} critical section in Section \hyperref[ReweightSoftScan]{3.5}, with one notable difference: the \texttt{hold-cluster} slot is populated with the hold cluster motivating our \texttt{multireweight}.
    This has the effect of treating the trees as if they were temporarily grafted onto a common root---the resulting sponsored action is guaranteed to be a \texttt{reweight}, with weight equal to half the distance between the closest trees in the cluster.%
    \footnote{%
Compare to the result in \Cref{PingPongAlgorithm} when two positive blossoms sharing the same root elect to \texttt{reweight}.
}
    \item
    Set the pingability of the hold cluster to \texttt{none}.
    \item
    Reweight the hold cluster by the weight from step 2.
    The supervisor interleaves the \texttt{reweight} operations for each tree in the cluster, which has the effect of reweighting the entire cluster simultaneously.
    \item
    Set the pingability of the hold cluster to \texttt{soft}.
    \item
    Perform another soft scan across the entire hold cluster, as in step 2, but no longer filling out the \texttt{hold-cluster} slot.
    This determines the minimum weight edge between the trees in the hold cluster and their surroundings.
    If this minimum edge weight is negative, rewind the hold cluster.
\end{enumerate}

% Once the supervisor has established a valid hold cluster, and if it was spawned by the root with the highest priority among the cluster, then the supervisor treats the trees as if they were temporarily grafted onto a common root.
% It performs a fresh \texttt{scan} over all the trees in the cluster with all of their roots included in the internal root set, which is then guaranteed to sponsor a reweight operation.%
% \footnote{%
% In \Cref{PingPongAlgorithm}, the blocking \texttt{hold} operations within the cluster become \text{pass}es, while the other reweighting rules maintain the feasibility of the modified edge weights.
% }
% The supervisor interleaves the steps in the reweight operations for each tree in the cluster, which has the effect of reweighting the entire cluster simultaneously, as illustrated in \Cref{AbutmentDeadlockFigure}.

\begin{example}
We demonstrate the \texttt{multireweight} operation in \Cref{AbutmentDeadlockFigure}. This operation is a consequence of constructing an alternating \emph{forest} rather than a lone alternating tree.
Starting from the first step of \Cref{AbutmentDeadlockFigure}, we illustrate how the serial algorithm would have progressed differently in \Cref{SerialAbutmentDeadlockFigure}.
\end{example}

\begin{figure}
% \begin{tikzpicture}[main/.style = {draw, circle}, node distance=15mm]
% \node[main] (r) {$r$};
% \node[main] (u) [right of=r] {$u, +1$};
% \node[main] (v) [right of=u] {$v$};
% \node[main] (s) [below of=r] {$s, +1$};
% \node[main] (t) [left of=s] {$t$};
% \node[main] (q) [right of=s] {$q$};
% \draw[densely dotted] (r) -- (u);
% \draw[densely dotted] (r) -- (s);
% \draw[densely dotted] (s) -- (q);
% \draw[densely dotted] (u) -- (q);
% \draw[very thick] (t) -- (s);
% \draw[very thick] (u) -- (v);
% \draw[densely dotted] (s) -- (u);
% \draw[dotted] (r) -- (q) node[midway, fill=white]{$2$};
% \draw[dotted] (r) -- (t) node[midway, fill=white]{$2$};
% \draw[dotted] (v) -- (q) node[midway, fill=white]{$2$};
% \end{tikzpicture}
% \[\Downarrow\]
\begin{tikzpicture}[main/.style = {draw, circle}, node distance=14mm]
\node[main] (r) {$r$};
\node[main] (u) [right of=r] {$u, +1$};
\node[main] (v) [right of=u] {$v$};
\node[main] (s) [below of=r] {$s, +1$};
\node[main] (t) [left of=s] {$t$};
\node[main] (q) [right of=s] {$q$};
\draw[->] (r) -- (u);
\draw[->] (r) -- (s);
\draw[densely dotted] (s) -- (q);
\draw[densely dotted] (u) -- (q);
\draw[->, very thick] (s) -- (t);
\draw[->, very thick] (u) -- (v);
\draw[densely dotted] (s) -- (u);
\draw[dotted] (r) -- (q) node[midway, fill=white]{$2$};
\draw[dotted] (r) -- (t) node[midway, fill=white]{$2$};
\draw[dotted] (v) -- (q) node[midway, fill=white]{$2$};
\addtikzpadding;
\end{tikzpicture}
\[\Downarrow\]
\begin{tikzpicture}[main/.style = {draw, circle}, node distance=14mm]
\node[main] (r) {$r, +1$};
\node[main] (u) [right of=r] {$u$};
\node[main] (v) [right of=u] {$v, +1$};
\node[main] (s) [below of=r] {$s$};
\node[main] (t) [left of=s] {$t, +1$};
\node[main] (q) [right of=s] {$q$};
\draw[->] (r) -- (u);
\draw[->] (r) -- (s);
\draw[dotted] (s) -- (q) node[midway, fill=white]{$1$};
\draw[dotted] (u) -- (q) node[midway, fill=white]{$1$};
\draw[->, very thick] (s) -- (t);
\draw[->, very thick] (u) -- (v);
\draw[dotted] (r) -- (q);
\draw[dotted] (s) -- (u) node[midway, fill=white]{$1$};
\draw[densely dotted] (r) -- (t);
\draw[dotted] (v) -- (q) node[pos=0.41, fill=white]{$1$};
\draw[densely dotted] (r) to [bend left] (v);
\end{tikzpicture}
\[\Downarrow\]
\begin{tikzpicture}[main/.style = {draw, circle}, node distance=14mm]
\node[main] (r) {$r, +1$};
\node[main] (u) [right of=r] {$u$};
\node[main] (v) [right of=u] {$v, +1$};
\node[main] (s) [below of=r] {$s$};
\node[main] (t) [left of=s] {$t, +1$};
\node[main] (q) [right=37mm of s] {$q, +1$};
\draw (r) to [bend right] (u);
\draw (r) to [bend left] (s);
\draw[densely dotted] (s) -- (q);
\draw[densely dotted] (u) -- (q);
\draw (s) to [bend left] (t);
\draw (u) to [bend right] (v);
\draw[very thick] (r) -- (q);
\draw (r) to [bend right] (t);
\draw[densely dotted] (v) -- (q);
\draw (r) to [bend left] (v);
\node[draw, inner sep=2mm, label=above left:$B_1$, fit=(r) (u) (v)] {};
\node[draw, inner sep=2mm, label=above left:$B_2$, fit=(r) (u) (v) (s) (t)] {};
\end{tikzpicture}
\[\Downarrow\]
% \begin{tikzpicture}[main/.style = {draw, circle}, node distance=15mm]
% \node[main] (r) {$r, +1$};
% \node[main] (u) [right of=r] {$u$};
% \node[main] (v) [right of=u] {$v, +1$};
% \node[main] (s) [below of=r] {$s$};
% \node[main] (t) [left of=s] {$t, +1$};
% \node[main] (q) [right of=s] {$q, +1$};
% \draw (r) to [bend right] (u);
% \draw[very thick] (s) -- (t);
% \draw (u) to [bend right] (v);
% \draw[very thick] (r) -- (q);
% \draw (r) to [bend left] (v);
% \node[red, draw, inner sep=2mm, label=above left:$V$, fit=(r) (u) (v)] {};
% \end{tikzpicture}
% \[\Downarrow\]
\begin{tikzpicture}[main/.style = {draw, circle}, node distance=14mm]
\node[main] (r) {$r, +1$};
\node[main] (u) [right of=r] {$u$};
\node[main] (v) [right of=u] {$v, +1$};
\node[main] (s) [below of=r] {$s$};
\node[main] (t) [left of=s] {$t, +1$};
\node[main] (q) [right of=s] {$q, +1$};
\draw[very thick] (s) -- (t);
\draw[very thick] (u) -- (v);
\draw[very thick] (r) -- (q);
\end{tikzpicture}
\caption{%
% The multireweight operation is a consequence of constructing an alternating \emph{forest} rather than a lone alternating tree.
% Starting from the partial matching constructed at the second step of \Cref{AbutmentDeadlockFigure}, we illustrate how the serial algorithm would have continued the matching, using $r$ as the only active root.
First, taking $r$ as the active root, $r$ grafts the two available match edges and reweights.
While $q$ is still disconnected in $G_\circ$, there are two cycles available from which it forms macrovertices $B_1$ and $B_2$.
In $G'$, $B_2$ (or, equivalently, $q$) can be reweighted, and $q$ can then form a match edge to any of $r$, $s$, $u$, and $v$.
Finally, expanding $B_2$ and $B_1$ lifts the perfect matching on $G'$ to one on $G$.
}
\label{SerialAbutmentDeadlockFigure}
\end{figure}
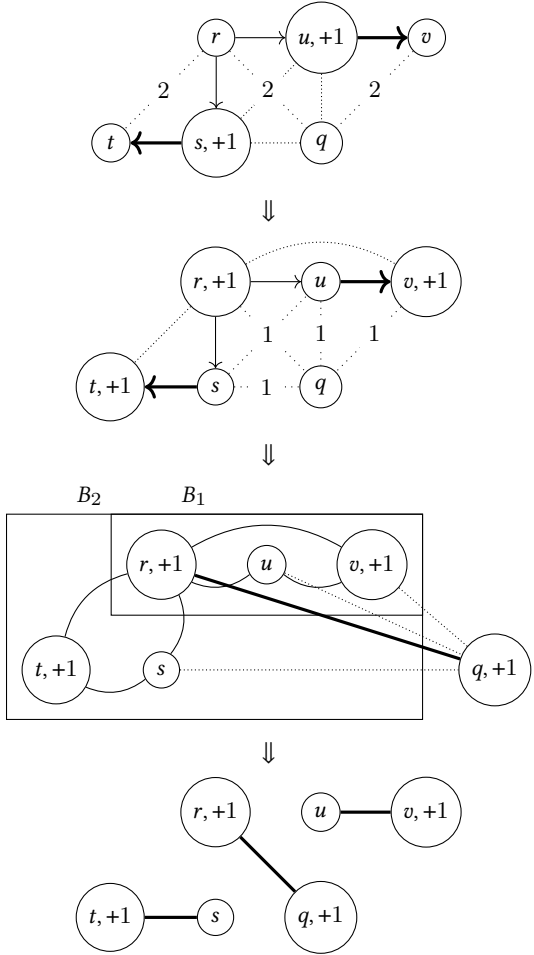

\subsection{The dryad main loop}

Finally, we consider the apparatus for starting and stopping the blossom algorithm.
The dryad, first mentioned in \Cref{DistributedEnvironmentSection} as a directory for the solver's various components, is a natural deposit for this responsibility.
The interface that it exposes between the solver and the outside world consists of the following messages:

\begin{description}
    \item[\texttt{message-sow}]  % external -> dryad
    Instructs the dryad to inject a new vertex into the solver.
    Carries an \texttt{id} which is used to refer uniquely to the vertex and to calculate edge weights in the graph.
    \item[\texttt{message-reap}]  % dryad -> external
    Sent by the dryad to \texttt{match-address} (a slot on the dryad) to announce a single matched edge in the solution, in the form of a pair of \texttt{id}s.
\end{description}

\begin{example}
We describe an implementation of a dryad with some major simplifying assumptions:
\begin{itemize}
    \item
    The dryad is ``monolithic'', meaning there is a single process responsible for servicing the API.
    \item
    The problem graph is complete, has an even number of vertices, and is fixed for the duration of the solver.
    Such a graph is guaranteed to have a perfect matching in which all vertices participate, simplifying the termination condition.
    \item
    Rather than provide a separate message which installs a plurality of vertices at once, we expect that the plurality of \texttt{message-sow}s arrive together.%
    \footnote{%
    This assumption has no effect  on the functionality of the algorithm if the entire problem is known at initialization time.
    }
\end{itemize}
On initialization, the dryad reads the pending \texttt{message-sow}s, spawns a process for each, and waits.
Each vertex process sends a \texttt{message-sprout} to the dryad when it joins the partial matching; once when every vertex makes that announcement, the solver has finished.
The dryad then queries the vertices for their match edges to announce the solution.
If any vertex replies that it does not have a match edge, it is because it is currently wrapped in a macrovertex.
Since the solver has finished, the dryad can safely send a directive for that macrovertex to expand, producing matches for all of its petals.
This eventually terminates, at which time the dryad announces the set of matched edges by emitting \texttt{message-reap}s.
\end{example}

\section{Analysis}

We now turn to the efficacy and efficiency of our variant of the blossom algorithm.

\subsection{Correctness}

As our algorithm is a modification of Edmonds's algorithm, our proof of correctness also hews closely to his, so we briefly recall the steps he takes.
He first analyzes the matching polytope so as to give a combinatorial recognition principle for minimum-weight maximum matchings~\cite[Section 4, Theorem (M)]{EdmondsMaximumMatching01}.
Namely, a matching $M = M_0$ on a graph $G = G_0$ qualifies when it can be extended to a sequence of graphs and matchings $\{M_j \subseteq G_j\}_{j=0}^n$ with the following properties:
\begin{itemize}
    \item Each $G_j$ is related to the next by a blossom contraction through an alternating cycle.
    \item Each $G_j$ is assigned ``admissible'' internal vertex weights.
    \item The vertex weights differ between $G_j$ and $G_{j+1}$ only at the new macrovertex.
    \item The contraction edges in each $G_j$ become weightless for these internal weights.
    \item All matched edges in $G_n$ are weightless, and all unmatched vertices in $G_n$ have zero internal weight.
\end{itemize}

To show these conditions suffice, Edmonds uses a polytope which encodes the maximum matching problem, with the following executive summary:

\begin{description}
    \item[{\cite[Section 2]{EdmondsMaximumMatching01}}]
    A matching gives rise to a vertex in this polytope, and a minimum-weight maximum matching is a polytope vertex which extremizes a linear functional encoding the edge weights.
    \item[{\cite[Section 5]{EdmondsMaximumMatching01}}]
    In general, polytope vertices can be modeled by sequences $\{M_j \subseteq G_j\}_j$ of the above type for \emph{some} choice of edge weights, hence are solutions to \emph{some} minimum-weight maximum matching problem.
    \item[{\cite[Section 6]{EdmondsMaximumMatching01}}]
    Given a polytope vertex extremizing a fixed edge weight functional, a sequence $\{M_j \subseteq G_j\}_j$ for that particular functional can be constructed.
    \item[{\cite[Section 7]{EdmondsMaximumMatching01}}]
    Finally, he provides a description of a primal-dual solver with verifiable progression toward such a combinatorial sequence---essentially, \Cref{SerialBlossomAlgorithm}.
\end{description}

\noindent
Since we are solving the same problem, we can reuse his recognition principle as-is.
In fact, we can reuse most of his proof that his algorithm produces the desired witnessing sequence of combinatorial steps: the sequence $\{M_j \subseteq G_j\}_j$ is read off from the steps in the pre-termination while loop from \Cref{SerialBlossomAlgorithm}, and with the same procedure we can read off such a sequence from the final state of the distributed algorithm.
We need only show that our algorithm makes progress toward this same goal.

\begin{figure}
\begin{center}
\begin{tikzpicture}[thick,
  every node/.style={draw,circle},
  mondenode/.style={fill=black!10!yellow},
  archnode/.style={fill=black!30!cyan},
  circletnode/.style={fill=black!30!orange},
  every fit/.style={ellipse,draw,inner sep=-2pt,text width=2cm},
  shorten >= 3pt,shorten <= 3pt,
  thin
]

\begin{scope}[node distance=30mm]
\node[circletnode] (circlet1) {};
\node[circletnode, left of=circlet1] (circlet2) {};
\node[circletnode, below right of=circlet1, yshift=15mm] (circlet3) {};
\node[circletnode, left of=circlet3] (circlet4) {};
\node[right of=circlet3, draw=none, node distance=10mm] {(circlet)};
\end{scope}

\begin{scope}[node distance=18mm]
\foreach \i in {1,2,3,4}
  \node[archnode, above of=circlet\i] (arch\i) {};
\node[right of=arch3, draw=none, node distance=10mm] {(arches)};
\end{scope}

\begin{scope}[node distance=18mm]
\foreach \i in {1,2,3,4}
  \node[mondenode, above of=arch\i] (monde\i) {};
\node[right of=monde3, draw=none, node distance=10mm] {(monde)};
\end{scope}

% the standard edges
\foreach \i in {1,2,3,4}
  \foreach \j in {1,2,3,4}
    \draw[densely dotted] (monde\i) -- (arch\j);

\foreach \i in {1,2,3,4}
  \draw[densely dotted] (circlet\i) -- (arch\i);

% edges in G
\draw[very thick] (circlet1) -- (circlet3);
\draw[densely dotted] (circlet3) -- (circlet4);
\draw[densely dotted] (circlet1) -- (circlet2);

% regal match edges
\draw[very thick] (circlet2) -- (arch2);
\draw[very thick] (circlet4) -- (arch4);
\draw[very thick] (arch1) -- (monde1);
\draw[very thick] (arch3) -- (monde3);
\addtikzpadding;
\end{tikzpicture}
\end{center}
\caption{
The crown on a four-vertex linear graph.
The bottom vertices form the (linear, four-vertex) circlet, the middle vertices and downward edges the arches, and the top vertices and upward edges the monde.
The matched edges depict a regal matching extending the matching consisting of just the rightmost circlet edge.
}\label{CrownExampleFigure}
\end{figure}
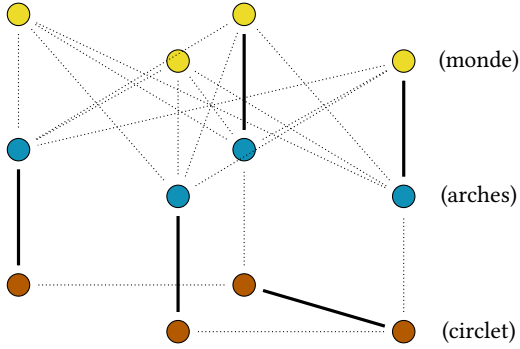

\subsubsection{Primal correctness}

For primal updates (i.e., for all steps save reweighting), we deduce correctness of the distributed algorithm by realizing its behavior as a special case of that of the serial algorithm.
Specifically, the primal updates in the distributed algorithm operating on a graph $G$ behave like the serial algorithm applied to a slightly larger graph $\Crown{G}$, which we call the \define{crown}.
We proceed as follows:
\begin{enumerate}
\item
We introduce normalizing conditions on the steps taken by the distributed and classical algorithms.
Essentially, a run of the distributed algorithm is said to be ``rooted'' when its tree operations do not interleave with one another, and a run of the classical algorithm is said to be ``good'' when it does not involve the new crown graph vertices ``too much''.
\item
We show that every run of either algorithm can be modified so that the respective normalization conditions hold: operations in the distributed algorithm can be judiciously delayed so as to become non-interleaved, and the classical algorithm can be discouraged from manipulating the new vertices in the crown graph without preventing progress.
\item
In the presence of the normalization conditions, we show that a rooted run of the distributed algorithm corresponds exactly to a good run of the classical algorithm.
We do this by directly comparing individual steps in one algorithm to individual steps in the other.
\item
We use this correspondence to transport Edmonds's witnessing sequence for a good run of the classical algorithm on $\Crown G$ to a witnessing sequence for a rooted run of the distributed algorithm on $G$.
Altogether, this establishes primal correctness.
\end{enumerate}

\begin{figure}[t!]
% \vspace{0.2cm}
\includegraphics[width=0.35\textwidth, trim={2cm 1.8cm 2cm 3cm}, clip]{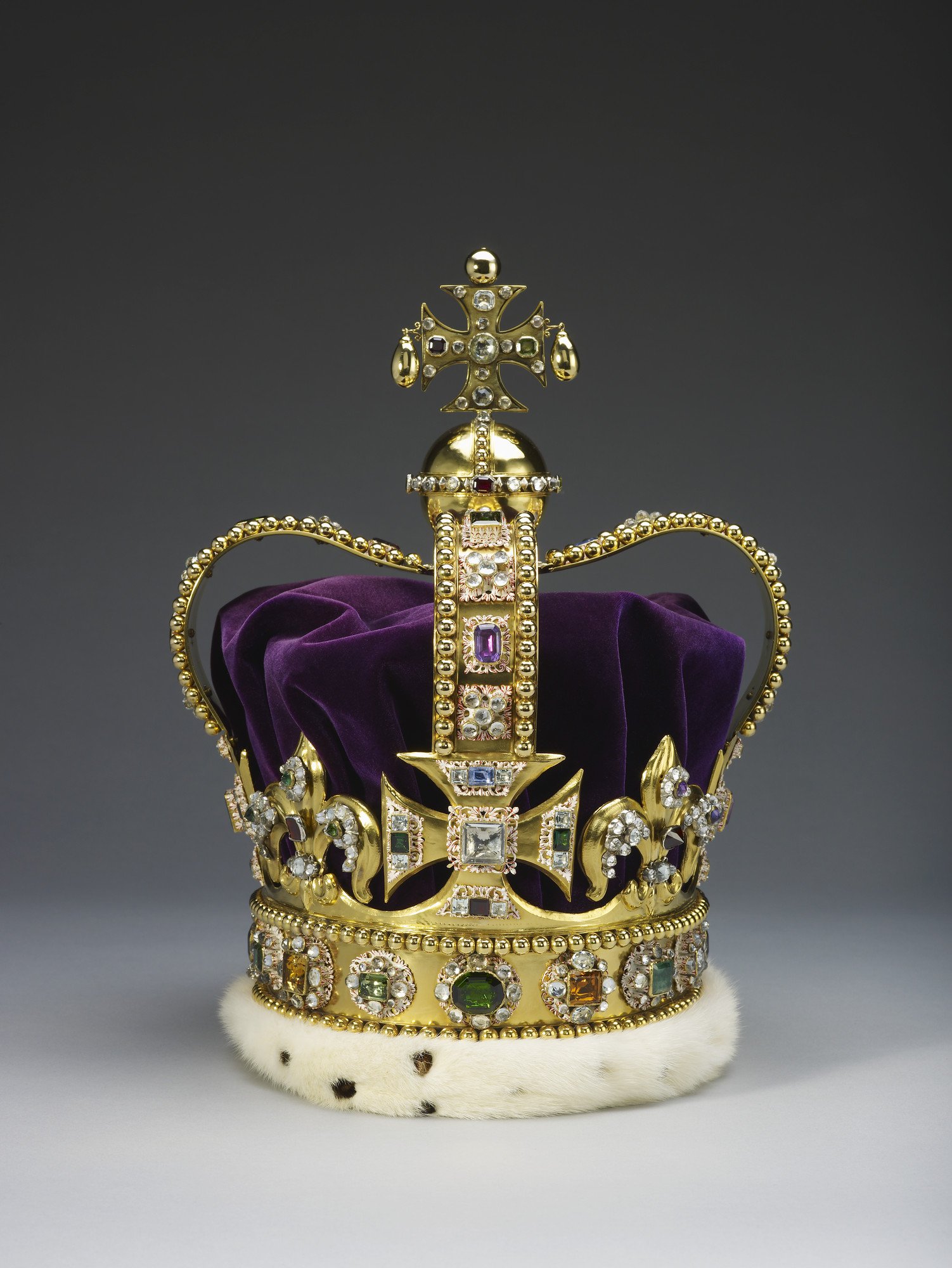}
\caption{St.\ Edward's crown~\cite{Crown} with spherical monde at top connected by arches to circlet at bottom.}\label{StEdwardsCrownFig}
\end{figure}

The idea behind the construction $\Crown{G}$ is to incorporate the distributed algorithm's supervisor processes as actual vertices in the problem graph.
This then enables us to identify the supervisor actions in the distributed algorithm with the behavior of trees rooted at these vertices in the serial algorithm.
The following definition makes this precise:

\begin{definition}\label{CrownDefn}
For an unweighted graph $G_\circ$, we define its \define{crown} $\Crown G_\circ$ as the following iterated pushout:
\begin{center}
\begin{tikzcd}[row sep=2mm]
& G^{\mathrm{disc}}_\circ \arrow{d} \arrow{r} & K_{G^{\mathrm{disc}}_\circ, G^{\mathrm{disc}}_\circ} \arrow{dd} \\
G^{\mathrm{disc}}_\circ \arrow{d} \arrow{r} & \operatorname{Cyl}(G^{\mathrm{disc}}_\circ) \arrow{rd} \\
G_\circ \arrow{rr} & & \Crown{G}_\circ,
\end{tikzcd}
\end{center}
where $G^{\mathrm{disc}}_\circ$ is the discretization, and \[\operatorname{Cyl}(-) = (-) \square \{0 \to 1\}\] is the \define{cylinder graph}, constructed using the graph ``Cartesian product''.
Concretely, the vertices of $\Crown G_\circ$ partition into three sets, each separately isomorphic to the vertices of $G_\circ$: the \define{circlet}, the \define{arches}, and the \define{monde}. For a vertex $v \in G_\circ$, we refer to the corresponding vertices in the circlet, arch, and monde as its \define{avatars}.
The edges joining these regions are as follows:
\begin{description}
  \item[Arches--arches:] None.
  \item[Arches--circlet:] Each pair of avatars of $v$ is joined by an edge.
  \item[Arches--monde:] Every arch vertex is joined to every monde vertex by an edge.
  \item[Circlet--circlet:] These are the same as the edges of $G_\circ$.
  \item[Circlet--monde:] None.
  \item[Monde--monde:] None.
\end{description}
\end{definition}

\begin{example}
\Cref{CrownExampleFigure} depicts the crown $\Crown G_\circ$ of a four-vertex linear graph $G_\circ$.
The circlet, arches, and monde are vertically arranged.
Additionally, we illustrate a regal matching on $\Crown G_\circ$ (see \Cref{RegalDefn}) which extends a partial matching on $G_\circ$ consisting of a single edge.
In \Cref{StEdwardsCrownFig} we provide the inspiration behind the naming of $\Crown G$, and
in \Cref{PushoutExampleFigure} we separately indicate how the crown is constructed as an iterated pushout.
\end{example}

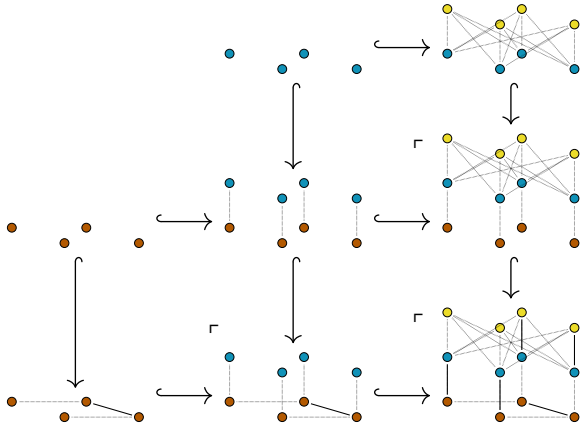
\begin{figure}
\begin{center}
\begin{tikzcd}
&
% G^disc arches
\scalebox{0.33}{
\begin{tikzpicture}[thick,
  every node/.style={draw,circle},
  mondenode/.style={fill=black!10!yellow},
  archnode/.style={fill=black!30!cyan},
  circletnode/.style={fill=black!30!orange},
  every fit/.style={ellipse,draw,inner sep=-2pt,text width=2cm},
  shorten >= 3pt,shorten <= 3pt,
  thin
]

\begin{scope}[node distance=30mm]
\node[archnode] (circlet1) {};
\node[archnode, left of=circlet1] (circlet2) {};
\node[archnode, below right of=circlet1, yshift=15mm] (circlet3) {};
\node[archnode, left of=circlet3] (circlet4) {};
\end{scope}
\end{tikzpicture}
}
\arrow[hook]{r}
\arrow[hook]{d}

&

% connected graph
\scalebox{0.33}{
\begin{tikzpicture}[thick,
  every node/.style={draw,circle},
  mondenode/.style={fill=black!10!yellow},
  archnode/.style={fill=black!30!cyan},
  circletnode/.style={fill=black!30!orange},
  every fit/.style={ellipse,draw,inner sep=-2pt,text width=2cm},
  shorten >= 3pt,shorten <= 3pt,
  thin
]

\begin{scope}[node distance=30mm]
\node[archnode] (circlet1) {};
\node[archnode, left of=circlet1] (circlet2) {};
\node[archnode, below right of=circlet1, yshift=15mm] (circlet3) {};
\node[archnode, left of=circlet3] (circlet4) {};
\end{scope}

\begin{scope}[node distance=18mm]
\foreach \i in {1,2,3,4}
  \node[mondenode, above of=circlet\i] (monde\i) {};
\end{scope}

% the standard edges
\foreach \i in {1,2,3,4}
  \foreach \j in {1,2,3,4}
    \draw[densely dotted] (monde\i) -- (circlet\j);
\end{tikzpicture}
}
\arrow[hook]{d}

\\

% G^disc circlet
\scalebox{0.33}{
\begin{tikzpicture}[thick,
  every node/.style={draw,circle},
  mondenode/.style={fill=black!10!yellow},
  archnode/.style={fill=black!30!cyan},
  circletnode/.style={fill=black!30!orange},
  every fit/.style={ellipse,draw,inner sep=-2pt,text width=2cm},
  shorten >= 3pt,shorten <= 3pt,
  thin
]

\begin{scope}[node distance=30mm]
\node[circletnode] (circlet1) {};
\node[circletnode, left of=circlet1] (circlet2) {};
\node[circletnode, below right of=circlet1, yshift=15mm] (circlet3) {};
\node[circletnode, left of=circlet3] (circlet4) {};
\end{scope}
\end{tikzpicture}
}
\arrow[hook]{r}
\arrow[hook]{d}

&

% Cyl(G^disc)
\scalebox{0.33}{
\begin{tikzpicture}[thick,
  every node/.style={draw,circle},
  mondenode/.style={fill=black!10!yellow},
  archnode/.style={fill=black!30!cyan},
  circletnode/.style={fill=black!30!orange},
  every fit/.style={ellipse,draw,inner sep=-2pt,text width=2cm},
  shorten >= 3pt,shorten <= 3pt,
  thin
]

\begin{scope}[node distance=30mm]
\node[circletnode] (circlet1) {};
\node[circletnode, left of=circlet1] (circlet2) {};
\node[circletnode, below right of=circlet1, yshift=15mm] (circlet3) {};
\node[circletnode, left of=circlet3] (circlet4) {};
\end{scope}

\begin{scope}[node distance=18mm]
\foreach \i in {1,2,3,4}
  \node[archnode, above of=circlet\i] (arch\i) {};
\end{scope}

\foreach \i in {1,2,3,4}
  \draw[densely dotted] (circlet\i) -- (arch\i);
\end{tikzpicture}
}
\arrow[hook]{r}
\arrow[hook]{d}

&

% upper right pushout
\scalebox{0.33}{
\begin{tikzpicture}[thick,
  every node/.style={draw,circle},
  mondenode/.style={fill=black!10!yellow},
  archnode/.style={fill=black!30!cyan},
  circletnode/.style={fill=black!30!orange},
  every fit/.style={ellipse,draw,inner sep=-2pt,text width=2cm},
  shorten >= 3pt,shorten <= 3pt,
  thin
]

\begin{scope}[node distance=30mm]
\node[circletnode] (circlet1) {};
\node[circletnode, left of=circlet1] (circlet2) {};
\node[circletnode, below right of=circlet1, yshift=15mm] (circlet3) {};
\node[circletnode, left of=circlet3] (circlet4) {};
\end{scope}

\begin{scope}[node distance=18mm]
\foreach \i in {1,2,3,4}
  \node[archnode, above of=circlet\i] (arch\i) {};
\end{scope}

\begin{scope}[node distance=18mm]
\foreach \i in {1,2,3,4}
  \node[mondenode, above of=arch\i] (monde\i) {};
\end{scope}

% the standard edges
\foreach \i in {1,2,3,4}
  \foreach \j in {1,2,3,4}
    \draw[densely dotted] (monde\i) -- (arch\j);

\foreach \i in {1,2,3,4}
  \draw[densely dotted] (circlet\i) -- (arch\i);
\end{tikzpicture}
}
\arrow[hook]{d}
\arrow[ul, phantom, "\ulcorner", very near start]

\\

% G
\scalebox{0.33}{
\begin{tikzpicture}[thick,
  every node/.style={draw,circle},
  mondenode/.style={fill=black!10!yellow},
  archnode/.style={fill=black!30!cyan},
  circletnode/.style={fill=black!30!orange},
  every fit/.style={ellipse,draw,inner sep=-2pt,text width=2cm},
  shorten >= 3pt,shorten <= 3pt,
  thin
]

\begin{scope}[node distance=30mm]
\node[circletnode] (circlet1) {};
\node[circletnode, left of=circlet1] (circlet2) {};
\node[circletnode, below right of=circlet1, yshift=15mm] (circlet3) {};
\node[circletnode, left of=circlet3] (circlet4) {};
\end{scope}

\draw[very thick] (circlet1) -- (circlet3);
\draw[densely dotted] (circlet3) -- (circlet4);
\draw[densely dotted] (circlet1) -- (circlet2);
\end{tikzpicture}
}
\arrow[hook]{r}

&

% lower left pushout
\scalebox{0.33}{
\begin{tikzpicture}[thick,
  every node/.style={draw,circle},
  mondenode/.style={fill=black!10!yellow},
  archnode/.style={fill=black!30!cyan},
  circletnode/.style={fill=black!30!orange},
  every fit/.style={ellipse,draw,inner sep=-2pt,text width=2cm},
  shorten >= 3pt,shorten <= 3pt,
  thin
]

\begin{scope}[node distance=30mm]
\node[circletnode] (circlet1) {};
\node[circletnode, left of=circlet1] (circlet2) {};
\node[circletnode, below right of=circlet1, yshift=15mm] (circlet3) {};
\node[circletnode, left of=circlet3] (circlet4) {};
\end{scope}

\begin{scope}[node distance=18mm]
\foreach \i in {1,2,3,4}
  \node[archnode, above of=circlet\i] (arch\i) {};
\end{scope}

% the standard edges
\foreach \i in {1,2,3,4}
  \draw[densely dotted] (circlet\i) -- (arch\i);

% edges in G
\draw[very thick] (circlet1) -- (circlet3);
\draw[densely dotted] (circlet3) -- (circlet4);
\draw[densely dotted] (circlet1) -- (circlet2);
\end{tikzpicture}
}
\arrow[hook]{r}
\arrow[ul, phantom, "\ulcorner", very near start]

&

% full picture
\scalebox{0.33}{
\begin{tikzpicture}[thick,
  every node/.style={draw,circle},
  mondenode/.style={fill=black!10!yellow},
  archnode/.style={fill=black!30!cyan},
  circletnode/.style={fill=black!30!orange},
  every fit/.style={ellipse,draw,inner sep=-2pt,text width=2cm},
  shorten >= 3pt,shorten <= 3pt,
  thin
]

\begin{scope}[node distance=30mm]
\node[circletnode] (circlet1) {};
\node[circletnode, left of=circlet1] (circlet2) {};
\node[circletnode, below right of=circlet1, yshift=15mm] (circlet3) {};
\node[circletnode, left of=circlet3] (circlet4) {};
\end{scope}

\begin{scope}[node distance=18mm]
\foreach \i in {1,2,3,4}
  \node[archnode, above of=circlet\i] (arch\i) {};
\end{scope}

\begin{scope}[node distance=18mm]
\foreach \i in {1,2,3,4}
  \node[mondenode, above of=arch\i] (monde\i) {};
\end{scope}

% the standard edges
\foreach \i in {1,2,3,4}
  \foreach \j in {1,2,3,4}
    \draw[densely dotted] (monde\i) -- (arch\j);

\foreach \i in {1,2,3,4}
  \draw[densely dotted] (circlet\i) -- (arch\i);

% edges in G
\draw[very thick] (circlet1) -- (circlet3);
\draw[densely dotted] (circlet3) -- (circlet4);
\draw[densely dotted] (circlet1) -- (circlet2);

% regal match edges
\draw[very thick] (circlet2) -- (arch2);
\draw[very thick] (circlet4) -- (arch4);
\draw[very thick] (arch1) -- (monde1);
\draw[very thick] (arch3) -- (monde3);
\end{tikzpicture}
}
\arrow[ul, phantom, "\ulcorner", very near start]
\end{tikzcd}
\end{center}
\caption{
The presentation of the crown on a four-vertex linear graph as an iterated pushout, including the two intermediate pushouts.
Each square describes its lower-right corner as two graphs, those in the upper-right and lower-left, glued together along the common subgraph in the upper-left.
}\label{PushoutExampleFigure}
\end{figure}

We show that these conditions permit the following correspondence, whose proof we briefly defer:

\begin{theorem}\label{UnweightedOrderBijection}
Rooted Lamport orderings of the distributed blossom algorithm acting on $G_\circ$ biject with good runs of the serial blossom algorithm acting on $\Crown G_\circ$, modulo the choice of monde vertex at the start of each serial segment.
This bijection preserves the sequence of state updates.
\end{theorem}

\noindent
To support this claim, we show in some Lemmas that these conditions are not so limiting.
Our first Lemma shows that by deferring irrelevant grafting operations, we can benignly reordering the events in a Lamport ordering so that the ordering becomes rooted.

To make precise the relationship between matchings on $\Crown{G}_\circ$ and matchings on $G_\circ$, we need three auxiliary definitions which help normalize the indeterminacy on both sides of our purported correspondence.
We begin with a condition on distributed runs:

\begin{definition}\label{RootedOrderingDefn}
In a Lamport ordering~\cite{Lamport} of the events of a run of the distributed algorithm, we refer to the sequence of grafting events between two adjacent state updates, as well as the later state update, as a \define{segment}.
A segment is said to be \define{rooted} if all of its roots belong to the set of roots which act in its final bookending state update.
The entire Lamport ordering is said to be \define{rooted} if all of its segments are rooted.
\end{definition}

\noindent
Our intent with \Cref{RootedOrderingDefn} is that rooted orderings display a kind of focused attention.
Rather than many tree operations happening in parallel in a forest, all the operations within a maximally reasonable block of time pertain to one tree only.

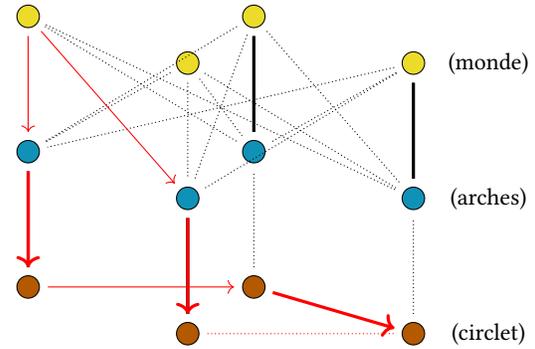
\begin{figure}
\begin{center}
\begin{tikzpicture}[thick,
  every node/.style={draw,circle},
  mondenode/.style={fill=black!10!yellow},
  archnode/.style={fill=black!30!cyan},
  circletnode/.style={fill=black!30!orange},
  every fit/.style={ellipse,draw,inner sep=-2pt,text width=2cm},
  shorten >= 3pt,shorten <= 3pt,
  thin
]

\begin{scope}[node distance=30mm]
\node[circletnode] (circlet1) {};
\node[circletnode, left of=circlet1] (circlet2) {};
\node[circletnode, below right of=circlet1, yshift=15mm] (circlet3) {};
\node[circletnode, left of=circlet3] (circlet4) {};
\node[right of=circlet3, draw=none, node distance=10mm] {(circlet)};
\end{scope}

\begin{scope}[node distance=18mm]
\foreach \i in {1,2,3,4}
  \node[archnode, above of=circlet\i] (arch\i) {};
\node[right of=arch3, draw=none, node distance=10mm] {(arches)};
\end{scope}

\begin{scope}[node distance=18mm]
\foreach \i in {1,2,3,4}
  \node[mondenode, above of=arch\i] (monde\i) {};
\node[right of=monde3, draw=none, node distance=10mm] {(monde)};
\end{scope}

% the standard edges
\foreach \i in {1,3,4}
  \foreach \j in {1,2,3,4}
    \draw[densely dotted] (monde\i) -- (arch\j);

\draw[densely dotted] (monde2) -- (arch1);
\draw[densely dotted] (monde2) -- (arch3);

\foreach \i in {1,2,3,4}
  \draw[densely dotted] (circlet\i) -- (arch\i);

% edges in G
% \draw[very thick] (circlet1) -- (circlet3);
% \draw[densely dotted] (circlet3) -- (circlet4);
% \draw[densely dotted] (circlet1) -- (circlet2);

% regal match edges
\draw[very thick] (circlet2) -- (arch2);
\draw[very thick] (circlet4) -- (arch4);
\draw[very thick] (arch1) -- (monde1);
\draw[very thick] (arch3) -- (monde3);

% alternating tree edges
\draw[red,->] (monde2) -- (arch2);
\draw[red,->,very thick] (arch2) -- (circlet2);
\draw[red,->] (circlet2) -- (circlet1);
\draw[red,->,very thick] (circlet1) -- (circlet3);
\draw[red,->] (monde2) -- (arch4);
\draw[red,->,very thick] (arch4) -- (circlet4);
\draw[red, densely dotted] (circlet4) -- (circlet3);
\addtikzpadding;
\end{tikzpicture}
\end{center}
\caption{
An alternating tree directed along the crown.
}\label{CrownExampleFigureWithTree}
\end{figure}

Next, we have conditions for the serial runs on the crown:

\begin{definition}\label{RegalDefn}
A matching on $\Crown G_\circ$ is said to be \define{regal} if all of the vertices in the circlet and arches are matched (but perhaps not to each other, i.e., some vertices in the arches may be matched to vertices in the monde).
The \define{initial regal matching} is the matching where each vertex in the circlet is matched to its avatar in the arches.
\end{definition}

\begin{definition}\label{GoodRunDefn}
A run of the serial algorithm on $\Crown G_\circ$ is said to be \define{good} if it satisfies the following properties:
\begin{itemize}
  \item The initial matching is regal.
  \item Every macrovertex formation in which a monde vertex participates is directly followed by an augment to another monde vertex and an expand operation on the macrovertex.
  \item Trees are \define{directed downwards along the crown}: No circlet vertex is permitted to graft an arch vertex as a child.
\end{itemize}
\end{definition}

\noindent
Our intent behind \Cref{RegalDefn} and \Cref{GoodRunDefn} is to limit the roles which the new monde and arch vertices can play in the serial algorithm, facilitating a comparison between the distributed algorithm acting directly on $G$ and the serial algorithm's behavior as understood through the circlet.

\begin{example}
In \Cref{CrownExampleFigureWithTree}, we then examine an example on the same $\Crown G_\circ$ of an alternating tree directed along the crown.
The two circlet vertices without circlet matches engender two unmatched crown vertices, at which the tree roots.
Restricted to the circlet (i.e., from the viewpoint of the distributed blossom algorithm), this tree has found an augmenting path between two unmatched circlet nodes, terminating with the foremost edge.
In the larger crown, this alternating tree is ready to contract into a macrovertex, at which point it can augment with the other unmatched monde vertex, then re-expand.
\end{example}

\begin{lemma}
Every run of the distributed algorithm admits a rooted Lamport ordering.
\end{lemma}
\begin{proof}
Fixing any Lamport ordering of a run, we inductively ``smooth'' it to produce a rooted ordering.
Consider the earliest segment within the run which is not rooted, and consider the last grafting operation within that segment which violates the rooted property.
Because the root which is enacting this graft is disjoint from all of the operations' roots from this point to the state update which terminates the segment, this graft operation can be commuted to occur just after the state update.
Continuing in this way produces the desired rooted ordering.
\end{proof}

\noindent
Next, we show that the \textit{good run} condition is not limiting: any partial good run which still admits progress as a run also admits progress as a good run.

\begin{lemma}
Every partial good run of the serial algorithm operating on $\Crown G_\circ$ can be extended to a longer good run.
\end{lemma}
\begin{proof}
We have two objectives:
\begin{enumerate}
  \item We must show that a contraction which includes a monde vertex can be followed by a monde-monde augmentation and an expansion.
  \item We must show that an augmenting path which does not respect the directedness property can be abbreviated to one that does.
\end{enumerate}

For the first objective, consider such a macrovertex and the cycle used to produce it.
The contracted cycle includes an arch vertex, hence the macrovertex inherits all of its edges to the vertices in the monde.
Any edge to an unmatched monde vertex can be used to perform the augmentation.
The macrovertex can then be directly expanded.

For the second objective, consider an augmenting path which violates directedness.
First, we may assume that the augmenting path visits the monde only at its endpoints as follows.
By regalness, the augmenting path must begin and end in the monde, so that any intermediate visit to the monde is through a matched vertex.
As the augmenting path begins and ends at positive vertices, there must be two adjacent monde vertices whose intervening path begins and ends with an unmatched edge.
Select that subpath, then replace its terminating edges with visits to the original terminating monde vertices.
Second, we may even assume that the terminating monde vertices are the same, since any arch vertex is reachable by any monde vertex.
Together, these observations produce an alternating cycle accessible by a crown-directed tree.
\end{proof}

\noindent
With these Lemmas in hand, we turn to the proof of the main correspondence between serial and distributed runs, stated as \Cref{UnweightedOrderBijection}.

\begin{proof}[{Proof of \Cref{UnweightedOrderBijection}}]
Given a rooted Lamport ordering of a run of the distributed blossom algorithm, select a segment for which we will construct a corresponding segment of steps in the serial algorithm.
Select an unmatched vertex from the monde to serve as the root of the serial segment.
For each root participating in the distributed segment, graft the corresponding circlet-arches matched edge to the chosen root in the monde.
Then, each non-augment maneuver in the distributed segment corresponds to an identical maneuver in the serial segment.
Finally, an augment in the distributed segment corresponds on the serial side to contraction, augmentation, and expansion.
Since we have assumed an interest only in good runs, this correspondence can also be applied in reverse.
\end{proof}

\begin{corollary}
Ignoring internal weights, the distributed algorithm can be used to produce the same $\{M_j \subseteq G_j\}_j$ data as in Edmonds's original proof~\cite[Section 7]{EdmondsMaximumMatching01}.
\qed
\end{corollary}

\subsubsection{Dual correctness}

Unfortunately, the dual step must be treated more manually: the multitude of monde vertices prevent us from fruitfully applying the weighted algorithm to $\Crown{G}$.
Fortunately, there is altogether less to show.
Instead, let us examine the the conditions for performing a reweight in the distributed setting.
\begin{description}
  \item[Positive--positive:]
  A tree for which no more primal updates are possible has no positive vertices connected by weightless edges to other positive vertices, whether in this tree or in another.
  Hence, this case is empty.
  This is as in the dual update step of the serial weighted algorithm.
  \item[Negative--any:]
  No edge emanating from a negative vertex in this tree affects the reweight calculation.
  This is also as in the dual update step of the serial weighted algorithm.
  \item[Positive--negative internal:]
  An edge emanating from a positive vertex to a negative vertex in the same tree is not included in the reweight amount calculation.
  This is also as in the dual update step of the serial weighted algorithm.
  \item[Positive--negative foreign:]
  An edge between a positive vertex in this tree and a negative vertex in another constrains the reweight amount.
  \textbf{Since the serial algorithm uses a tree rather than a forest, it has no concept of a ``foreign'' target and hence has no analogue of this case.}
\end{description}

\noindent
Our observation is that this last constraint is always a phantom: even when it is weightless, (non-local) progress can always be made.
There are two subcases to this final point:
according to whether the tree in question belongs to a hold cluster.
\begin{description}
  \item[Tree belongs to a hold cluster:]
  The multireweighting procedure circumnavigates the weightless hold constraints.
  In fact, it is guaranteed to reweight by a (nonzero) amount giving a minimal enlargement of the weightless subgraph.
  \item[Tree does not belong to a hold cluster:]
  The tree responsible for breaking the hold cluster is necessarily free to act in some other way, whether by a primal update or by a (multi)reweight of its own.
\end{description}
In Edmonds's language, both subcases (hence all cases) afford algorithmic progress.

Since the correctness of the distributed blossom algorithm comes down only to the combinatorial properties of the primal step, together with the ability of the second step to make progress, we have therefore proven the following theorem:

\begin{theorem}\label{DistributedCorrectness}
The distributed blossom algorithm is correct: it always terminates, and upon termination it emits a minimum-weight perfect matching.
\qed
\end{theorem}

\subsection{Timing}

\begin{theorem}\label{TimingTheorem}
The runtime of the distributed blossom algorithm is $O(n^4)$, with $n$ the number of vertices in $G$.
\end{theorem}
\begin{proof}
The serial algorithm does not backtrack, and using a priority mechanism the distributed algorithm \emph{also} does not backtrack.
Hence, the runtime of the distributed algorithm is bounded by the duration of the serial run to which it corresponds under \Cref{UnweightedOrderBijection}.
Since the original algorithm runs in time $O(n^4)$ and $\Crown{G}$ is not asymptotically larger than $G$, we learn the same for the distributed algorithm.
\end{proof}

\begin{remark}
In particular, there is never an asymptotic penalty to using this algorithm over the (original) serial version.
Of course, modern implementations of the serial algorithm reduce the asymptotic runtime substantially compared to Edmonds's original.
\end{remark}

\section{Closing comments}

We point out some of the stones we have left unturned.

\begin{remark}[Structured data improvements to runtime]
The worst case time complexity bound given in \Cref{TimingTheorem} is likely to be achievable for the algorithm presented here: when working with a fully-connected graph, reweighting operations rooted at high-priority nodes have a particularly bothersome ability to prevent lower-priority reweighting operations from succeeding.
However, improvements to the serial blossom algorithm have appeared since its invention in the 1960s~\cite{CookRohe,Kolmogorov,MicaliVazirani} which more intelligently deploy data structures to track which proposals are worth querying, lowering the time complexity bound.
It is open whether these serial blossom variants can be imported to the distributed setting, where the maintenance of delicate large-scale structure has the potential to destroy locality.
\end{remark}

\begin{remark}[Geometric improvements to runtime]\label{GeometricImprovements}
Our intended application rests heavily on the observation that geometric structure in the problem graphs can be leveraged in the dryad to give dramatic improvements in runtime~\cite{PetersonKaralekasArtemisPaper}.
A study of this phenomenon is well worth pursuing, not least because of the further application to quantum error correction.
\end{remark}

\begin{remark}[Lower bounds on runtime]
It is also of interest to find classes of problem graphs which cannot be quickly matched, thereby putting lower bounds on the runtime of the distributed algorithm.
For instance, it is a well-known result that $2$--coloring a line takes $o(n)$ time in the size of the line~\cite{Linial}, which is a special case of a perfect matching: given an enumeration of the vertices in the line (with no regard for their ordering) and a perfect matching, one can produce a $2$--coloring by partitioning vertices into those whose matches are higher- or lower-valued than they are.
Producing a family of examples of this form, as well as understanding their interactions with \Cref{GeometricImprovements}, would be very valuable.
\end{remark}

\begin{remark}[Online variants]
In our intended application, we implement a dryad that supports, to a limited but extremely useful extent, injection and ejection of nodes from an actively running algorithm.
It is of interest to understand how far this can be pushed and what applications this unlocks.
\end{remark}

\section*{Acknowledgements}

We would like to thank both Charles Hadfield and Austin Fowler for indirectly suggesting this problem: we would not have thought to look in this direction without Charles, nor would we have known where to begin without Fowler's groundwork.
We would also like to thank Zac Cranko for pointing out that it was possible to make a time complexity argument.
Lastly, the first author would like to thank Samrita Dhindsa for all manner of support as this work was carried out.

\bibliographystyle{alpha}
\bibliography{main}

\newcommand{\etalchar}[1]{$^{#1}$}
\begin{thebibliography}{CNAA{\etalchar{+}}20}

\bibitem[aet]{aether}
\texttt{aether} - {D}istributed system emulation in {C}ommon {L}isp.
\newblock \url{https://github.com/dtqec/aether}.
\newblock Accessed: 2021-01-11.

\bibitem[Ale49]{Aleichem}
Sholem Aleichem.
\newblock {\em {T}evye and his {D}aughters}.
\newblock Crown Publishers, 1949.

\bibitem[ana]{anatevka}
\texttt{anatevka} - {A} distributed blossom algorithm.
\newblock \url{https://github.com/dtqec/anatevka}.
\newblock Accessed: 2021-01-11.

\bibitem[Ash75]{ASHCROFT1975110}
E.A. Ashcroft.
\newblock Proving assertions about parallel programs.
\newblock {\em Journal of Computer and System Sciences}, 10(1):110--135, 1975.

\bibitem[CCPS09]{CCPS}
William~J Cook, WH~Cunningham, WR~Pulleyblank, and A~Schrijver.
\newblock {\em Combinatorial optimization}, volume~5.
\newblock Springer, 2009.

\bibitem[CDF{\etalchar{+}}19]{CDGFRR}
Armando Casta{\~{n}}eda, Carole Delporte{-}Gallet, Hugues Fauconnier, Sergio
  Rajsbaum, and Michel Raynal.
\newblock Making local algorithms wait-free: the case of ring coloring.
\newblock {\em Theory Comput. Syst.}, 63(2):344--365, 2019.

\bibitem[CNAA{\etalchar{+}}20]{CNAACHIPBFKRPJSNPB}
Christopher Chamberland, Kyungjoo Noh, Patricio Arrangoiz-Arriola, Earl~T.
  Campbell, Connor~T. Hann, Joseph Iverson, Harald Putterman, Thomas~C.
  Bohdanowicz, Steven~T. Flammia, Andrew Keller, Gil Refael, John Preskill,
  Liang Jiang, Amir~H. Safavi-Naeini, Oskar Painter, and Fernando G. S.~L.
  Brandão.
\newblock Building a fault-tolerant quantum computer using concatenated cat
  codes, 2020.

\bibitem[CR99]{CookRohe}
William~J. Cook and Andr{\'{e}} Rohe.
\newblock Computing minimum-weight perfect matchings.
\newblock {\em {INFORMS} J. Comput.}, 11(2):138--148, 1999.

\bibitem[DFF56]{DFF}
George~Bernard Dantzig, Lester~R Ford, and Delbert~Ray Fulkerson.
\newblock {\em A Primal-Dual Algorithm for Linear Programs}.
\newblock Princeton University Press, 1956.

\bibitem[DFFR19]{DGFFR}
Carole Delporte{-}Gallet, Hugues Fauconnier, Pierre Fraigniaud, and
  Mika{\"{e}}l Rabie.
\newblock Distributed computing in the asynchronous {LOCAL} model.
\newblock {\em CoRR}, abs/1904.07664, 2019.

\bibitem[DKLP02]{DKLP}
Eric Dennis, Alexei Kitaev, Andrew Landahl, and John Preskill.
\newblock Topological quantum memory.
\newblock {\em Journal of Mathematical Physics}, 43(9):4452--4505, 2002.

\bibitem[Edm65a]{EdmondsMaximumMatching01}
Jack Edmonds.
\newblock Maximum matching and a polyhedron with 0, 1-vertices.
\newblock {\em Journal of research of the National Bureau of Standards B},
  69(125-130):55--56, 1965.

\bibitem[Edm65b]{EdmondsPathsTreesFlowers}
Jack Edmonds.
\newblock Paths, trees, and flowers.
\newblock {\em Canadian Journal of mathematics}, 17:449--467, 1965.

\bibitem[FMMC12]{FMM}
Austin~G. Fowler, Matteo Mariantoni, John~M. Martinis, and Andrew~N. Cleland.
\newblock Surface codes: Towards practical large-scale quantum computation.
\newblock {\em Phys. Rev. A}, 86:032324, Sep 2012.

\bibitem[FWH12]{FWH}
Austin~G. Fowler, Adam~C. Whiteside, and Lloyd C.~L. Hollenberg.
\newblock Towards practical classical processing for the surface code: Timing
  analysis.
\newblock {\em Phys. Rev. A}, 86:042313, Oct 2012.

\bibitem[HBS73]{HBS}
Carl Hewitt, Peter~Boehler Bishop, and Richard Steiger.
\newblock A universal modular {ACTOR} formalism for artificial intelligence.
\newblock In Nils~J. Nilsson, editor, {\em Proceedings of the 3rd International
  Joint Conference on Artificial Intelligence. Standford, CA, USA, August
  20-23, 1973}, pages 235--245. William Kaufmann, 1973.

\bibitem[Kol09]{Kolmogorov}
Vladimir Kolmogorov.
\newblock Blossom {V:} a new implementation of a minimum cost perfect matching
  algorithm.
\newblock {\em Math. Program. Comput.}, 1(1):43--67, 2009.

\bibitem[Kuh10]{Kuhn}
Harold~W. Kuhn.
\newblock The hungarian method for the assignment problem.
\newblock In Michael J{\"{u}}nger, Thomas~M. Liebling, Denis Naddef, George~L.
  Nemhauser, William~R. Pulleyblank, Gerhard Reinelt, Giovanni Rinaldi, and
  Laurence~A. Wolsey, editors, {\em 50 Years of Integer Programming 1958-2008 -
  From the Early Years to the State-of-the-Art}, pages 29--47. Springer, 2010.

\bibitem[Lam78]{Lamport}
Leslie Lamport.
\newblock Time, clocks, and the ordering of events in a distributed system.
\newblock {\em Commun. {ACM}}, 21(7):558--565, 1978.

\bibitem[Lin92]{Linial}
Nathan Linial.
\newblock Locality in distributed graph algorithms.
\newblock {\em {SIAM} J. Comput.}, 21(1):193--201, 1992.

\bibitem[LPP15]{LPSP}
Zvi Lotker, Boaz Patt{-}Shamir, and Seth Pettie.
\newblock Improved distributed approximate matching.
\newblock {\em J. {ACM}}, 62(5):38:1--38:17, 2015.

\bibitem[LPR09]{LPSR}
Zvi Lotker, Boaz Patt{-}Shamir, and Adi Ros{\'{e}}n.
\newblock Distributed approximate matching.
\newblock {\em {SIAM} J. Comput.}, 39(2):445--460, 2009.

\bibitem[MV80]{MicaliVazirani}
Silvio Micali and Vijay~V. Vazirani.
\newblock An o(sqrt({\(\vert\)}v{\(\vert\)}) {\(\vert\)}e{\(\vert\)}) algorithm
  for finding maximum matching in general graphs.
\newblock In {\em 21st Annual Symposium on Foundations of Computer Science,
  Syracuse, New York, USA, 13-15 October 1980}, pages 17--27. {IEEE} Computer
  Society, 1980.

\bibitem[PK]{PetersonKaralekasArtemisPaper}
Eric~C. Peterson and Peter~J. Karalekas.
\newblock Distributed online decoding for large-scale quantum computers.
\newblock To appear.

\bibitem[PK20]{PetersonKaralekasAetherPaper}
Eric~C. Peterson and Peter~J. Karalekas.
\newblock aether: Distributed system emulation in common lisp.
\newblock {\em CoRR}, abs/2011.06180, 2020.

\bibitem[SRFR19]{Crown}
{S}gt.\ {R}upert~{F}rere {RLC/MOD}.
\newblock Photo of {S}t.\ {E}dward's {C}rown, September 2019.
\newblock
  \url{https://commons.m.wikimedia.org/wiki/File:St_Edward%27s_Crown_-_Royal_Collection_Trust_(1).jpg}.

\bibitem[WW04]{WattenhoferWattenhofer}
Mirjam Wattenhofer and Roger Wattenhofer.
\newblock Distributed weighted matching.
\newblock In Rachid Guerraoui, editor, {\em Distributed Computing, 18th
  International Conference, {DISC} 2004, Amsterdam, The Netherlands, October
  4-7, 2004, Proceedings}, volume 3274 of {\em Lecture Notes in Computer
  Science}, pages 335--348. Springer, 2004.

\end{thebibliography}

\end{document}